\documentclass[12pt, draftclsnofoot, onecolumn]{IEEEtran}
\IEEEoverridecommandlockouts    
\usepackage{cite}
\usepackage{amsmath,amssymb,amsfonts}
\usepackage{cite}
\usepackage{graphicx}
\usepackage{textcomp}
\usepackage{acronym}
\usepackage{xcolor}
\usepackage{subcaption} 
\usepackage{lipsum}     
\usepackage{algpseudocode}
\usepackage{lettrine}
\usepackage{amsthm}
\usepackage{tabularx}  
\newtheorem{lemma}{Lemma}
\usepackage{subfig}
\usepackage{subcaption}
\usepackage{algorithm}
\usepackage{amsmath}
\usepackage{algpseudocode}
\begin{document}
\acrodef{THz}[THz]{terahertz} 
\acrodef{THz-AP}[THz-AP]{terahertz-access point} 
\acrodef{THz-APs}[THz-APs]{terahertz-access points} 
\acrodef{VLC}[VLC]
{visible light communication}
\acrodef{VLC-AP}[VLC-AP]{visible light communication-access point} 
\acrodef{VLC-APs}[VLC-APs]{visible light communication-access points} 
\acrodef{APs}[APs]{access points}
\acrodef{SNR}[SNR]{signal-to-noise ratio}
\acrodef{$THz_s-AP$}[$THz_s-AP$]{terahertz sensing-access point}
\acrodef{$THz_c-AP$}[$THz_c-AP$]{terahertz communication-access point}
\acrodef{$VLC_c-AP$}[$VLC_c-AP$]{visible light communication-access point}
\acrodef{$VLC_c-APs$}[$VLC_c-APs$]{visible light communication-access points}
\acrodef{number of users}[number of users]{number of users}
\acrodef{$THz_c/VLC_c-AP$}[$THz_c/VLC_c-AP$]{terahertz/visible light communication-access point}
\acrodef{$THz_c$}[$THz_c$]{terahertz communication}
\acrodef{AP}[AP]{access point}
\acrodef{NC}[NC]{network controller}
\acrodef{$P_d$}[$P_d$]{probability of detection}
\acrodef{$FA_p$}[$FA_p$]{probability of false alarm}
\acrodef{$SC_p$}[$SC_p$]{sensing coverage probability}
\acrodef{$THz_s$}[$THz_s$]{terahertz sensing}
\acrodef{$THz_c$}[$THz_c$]{terahertz communication}
\acrodef{$VLC_c$}[$VLC_c$]{visible light communication}
\acrodef{ISAC}[ISAC]{integrated sensing and communication}
\acrodef{LoS}[LoS]{line-of-sight}
\acrodef{NLoS}[NLoS]{non line-of-sight}
\acrodef{SNR}[SNR]{signal-to-noise ratio}
\acrodef{SE}[SE]{spectral efficiency}
\acrodef{QoS}[QoS]{quality of service}
\acrodef{PSD}[PSD]{power spectral density}
\acrodef{SINR}[SINR]{signal-to-interference noise ratio}
\acrodef{$THz_c/VLC_c$}[$THz_c/VLC_c$]{terahertz/visible light communication}
\acrodef{MHCP-II}[MHCP-II]{matérn hard-core point process of type II} 
\acrodef{THz/VLC}[THz/VLC]{terahertz communication/visible light communication}
\acrodef{EE}[EE]{energy efficiency}

\title{Energy-Efficient THz Sensing with Hybrid THz/VLC Communication Under Human Blockage Effects}

\author{Hanshita Prabhakar, Neetu R.R, and Vivek Ashok Bohara
\thanks{The authors are from Indraprastha Institute of Information Technology Delhi (IIIT-D), New Delhi, India. email: \{hanshitap, neetur, vivek.b\}@iiitd.ac.in.}}
\maketitle

\maketitle
\vspace{-2cm}
\begin{abstract}
This paper presents an energy-efficient indoor system integrating \ac{THz} with \ac{VLC}. \ac{THz} communication offers ultra-high-capacity links but is limited by severe path loss, atmospheric absorption, and susceptibility to blockages. In contrast, \ac{VLC} provides robust, wide indoor coverage with illumination support, thereby enabling reliable, high-speed hybrid connectivity. To leverage their respective strengths, we propose a hybrid framework that integrates \ac{$THz_s-AP$} with hybrid \ac{$THz_c/VLC_c-AP$}, enabling reliable coverage and enhancing the \ac{EE} from an \ac{ISAC} perspective. We first perform optimal power allocation between the \ac{$THz_s-AP$} and \ac{$THz_c-AP$} to optimized the set of users served by the \ac{$THz_c-AP$} link, considering monostatic sensing performance metrics such as \ac{$P_d$}, \ac{$FA_p$} and \ac{$SC_p$} under the impact of human blockages are evaluated. Subsequently, the overall network power consumption is minimized via a mixed-integer linear programming (MILP) optimization that optimally selects the active \ac{$VLC_c-APs$} and assigns transmit powers.  Furthermore, extensive performance evaluations are conducted to analyze key metrics, including average energy efficiency, average spectral efficiency, average sensing rate, and average communication rate. Simulation results demonstrate that, under \ac{THz} sensing, most users are connected to the \ac{$THz_c-AP$} in the absence of blockages, whereas in the presence of blockages, the majority are served by the \ac{$VLC_c-APs$}. Overall, all users maintain reliable coverage with high \ac{EE}.
\end{abstract}

\vspace{-0.5cm}
\begin{IEEEkeywords}
\vspace{-0.3cm}
THz, VLC, MHCP-II, power allocation, ISAC, monostatic sensing, probability of detection, probability of false alarm, sensing coverage probability, power minimization, energy efficiency, spectral efficiency.
\end{IEEEkeywords}

\vspace{-0.5cm}
\section{Introduction}
 \IEEEPARstart{U}{nlike} fifth-generation (5G) networks that primarily operate in sub-6 GHz and millimeter (mm) bands, future wireless systems may extend into higher frequencies, notably the terahertz (\ac{THz}) band from 0.1–10 \ac{THz} and visible light communication (\ac{VLC}) spectra, alongside sub-300 GHz bands. This evolution enables multi-band architectures, where users with multi-band capable devices can simultaneously exploit diverse bands to enhance coverage, reliability, data rates, security, and energy efficiency (EE). The \ac{THz} band offers ultra-wide bandwidth, highly directional beams, and inherent link-level security, but suffers from severe path loss, susceptibility to blockage, and frequency-dependent molecular absorption. Conversely, \ac{VLC}, leveraging light sources for both illumination and communication, provides large bandwidth, immunity to electromagnetic interference, improved physical-layer security, and nil electromagnetic field (EMF) exposure, yet faces limitations from blockages, propagation losses, and device orientation\cite{10707308}. Together, \ac{THz} and \ac{VLC} serve as complementary enablers to sub-6 GHz and mmWave systems, supporting disruptive sixth-generation (6G) and beyond applications. Furthermore, their integration within multi-band architectures can enable seamless indoor–outdoor connectivity, dynamic spectrum utilization, and novel integrated sensing and communication (\ac{ISAC}) driven services. These advancements position \ac{THz} and \ac{VLC} as critical pillars in realizing the performance targets of next-generation wireless networks.

\ac{ISAC} has emerged as a fundamental enabler of 6G networks and is identified by both the International Telecommunication Union (ITU) \cite{ITU2023} and the Third Generation Partnership Project (3GPP) \cite{3GPP2023} as one of the six key use cases. \ac{ISAC} encompasses the joint design and implementation of communication and sensing functionalities within a unified system architecture, where a single base station is intrinsically capable of supporting both operations. In this framework, network sensing refers to the ability of the infrastructure to detect, localize, and characterize objects through parameters such as presence, shape, size, position, and velocity by exploiting transmitted and received radio signals Moreover, sensor fusion will be essential in delivering holistic sensing capabilities by combining network-derived data with information from heterogeneous sources, including environmental sensors (e.g., humidity and rainfall), localization tags, and device-embedded sensors such as cameras, light detection and ranging (LiDAR), and inertial measurement units.
 
\vspace{-0.4cm}
\subsection{Related Work}
In \cite{10711865}, the authors proposed a THz network placement optimization framework to ensure uniform data rate coverage in a hybrid \ac{THz/VLC} indoor environment, accounting for existing \ac{VLC-AP} positions and indoor layouts. The framework achieves locally optimal solutions with low complexity and is extended to sum rate maximization as a special case. Furthermore, the authors formulate a joint optimization of \ac{AP} assignment, subcarrier allocation, and power allocation to maximize handover-aware \ac{EE} under practical constraints, including blockages, interference, and \ac{QoS} requirements. Further, in \cite{9745789}, the authors have considered a hybrid \ac{THz/VLC} indoor wireless network, where \ac{THz-APs} deliver high-quality virtual reality content to users via THz links, while \ac{VLC-APs} provide precise indoor positioning services through \ac{VLC} links. A meta-reinforcement learning-based approach is employed to jointly optimize \ac{VLC-APs} selection for localization and \ac{THz-APs} association for data transmission, with the objective of maximizing overall network reliability. In \cite{karoti2022improving}, and \cite{10041414}, the authors propose an \ac{SINR} based link aggregation (LA) framework to improve heterogeneous light fidelity/wireless fidelity (LiFi/WiFi) network performance. The authors showed that dynamically distributing traffic based on real-time \ac{SINR} enhances user coverage, load balancing, and \ac{QoS}; the results were plotted for coverage probability, outage probability, and data rate. The authors in \cite{10757556} investigated the hybrid \ac{THz/VLC} communication system for an outdoor scenario and evaluated their performance through detailed bit error rate analysis. For the indoor \ac{THz} sensing, the authors in \cite{10622748},\cite{11059649},\cite{7562539}, and \cite{10437749} analyzed the sensing coverage probability, throughput, and sensing rate, providing a detailed trade-off between the average sensing rate and communication rate and the hardware implementations. \cite{10881940}, investigates indoor office (28, 60, 77 GHz) and outdoor vehicular (60, 77, 94 GHz) scenarios under object-only and object-plus-clutter cases. Results show that clutter impacts range doppler resolution, while higher frequencies improve sensing performance at the cost of increased path loss. \cite{behdad2022power} and \cite{10058989} investigate cell-free massive multiple-input output (MIMO) \ac{ISAC} systems, where \ac{APs} jointly serve users and detect targets using either reused downlink symbols or dedicated sensing symbols. A two-step design maximizes sensing \ac{SNR} and applies maximum a posteriori ratio test (MAPRT) detection, showing improved performance with dedicated resources. Moreover, it turns beam squint and beam split in mmWave/\ac{THz} MIMO into an advantage by using true time delay (TTD) lines for rapid, parallel direction finding. Together, these methods demonstrate \ac{ISAC} potential for efficient, reliable, and spectrum-aware joint communication and sensing in future networks. \cite{10439221}, offers an in-depth physical-layer perspective on THz-enabled \ac{ISAC}, laying out the fundamental challenges, key signal processing components, performance analysis, and future research directions necessary for realizing integrated sensing and communication in 6G and beyond. Furthermore, \cite{10633859},\cite{10438962}, and \cite{9737357} introduce a novel \ac{ISAC}-based beam alignment framework for \ac{THz} networks to tackle beam misalignment caused by \ac{LoS} blockages and user mobility, which severely impact coverage. By integrating advanced sensing into beam management, the approach enables intelligent environment-aware beam direction, significantly reducing misalignment and improving coverage. Moreover, it is fully compatible with existing 5G beam management architectures, ensuring both reliability and practical deployability in future THz networks. In \cite{shen2024energy} and \cite{zhang2020energy} focus on enhancing \ac{EE} in \ac{THz} communication systems through advanced resource allocation and beamforming strategies. One optimizes power allocation and user pairing in THz downlink non-orthogonal multiple access (NOMA) to boost \ac{SE} and \ac{EE}, while the other minimizes beam training overhead to reduce energy consumption without sacrificing alignment accuracy. Together, they demonstrate significant gains in EE, throughput, and latency over conventional schemes.

\vspace{-0.4cm}
\subsection{Motivation and Contribution}
{Prior research on hybrid VLC/radio frequency (RF) \cite{karoti2022improving}, \cite{khreishah2018hybrid} systems and hybrid \ac{THz}/\ac{VLC} \cite{10757556},\cite{10711865} architectures has primarily focused on communication performance, optimization strategies, and handover-aware \ac{EE}. The hybrid VLC/RF studies \cite{karoti2022improving}, \cite{khreishah2018hybrid} have provided solutions for moderate data-rate coverage and power optimization; however, the impact of blockage scenarios has largely been overlooked. Similarly, hybrid THz/VLC works \cite{10757556},\cite{10711865}  have mainly analyzed system performance in terms of user sum rates and coverage, while the effects of blockages, joint power optimization, and the sensing capabilities enabled by THz access points remain insufficiently explored. Consequently, integrating THz sensing functionality with hybrid THz/VLC communication under blockage conditions while jointly minimizing power consumption remains a largely unexplored research direction. Although the above studies highlight the potential of hybrid technologies, they are predominantly communication-centric. A key challenge lies in ensuring reliable coverage, high data rates, and \ac{EE}, as both THz and VLC systems suffer from limited coverage, severe path loss, and high sensitivity to human-induced blockages, particularly in dense indoor environments. The incorporation of sensing into hybrid THz/VLC systems is still limited. Moreover, the impact of blockages on THz sensing, both independently and in conjunction with hybrid THz/VLC communication, and its implications for power consumption and overall \ac{EE} have not been comprehensively investigated in the existing literature, especially from a sensing-oriented perspective.}

{This work focuses on reducing the power consumption of an indoor wireless access network in the presence of human-induced blockages. Unlike existing studies, we propose a joint \ac{$THz_s-AP$} sensing with a hybrid \ac{$THz_c/VLC_c-AP$} communication. By incorporating multiple \ac{$VLC_c-APs$} alongside a \ac{$THz_c-AP$}, the proposed system leverages complementary coverage to effectively mitigate blockage effects.
As a result, users can be supported with ultra-high data-rate applications, while simultaneously accommodating moderate data-rate services. Furthermore, under severe blockage conditions, connectivity can be maintained through \ac{$VLC_c-APs$}, enabling visible light assisted communication.
The primary challenge addressed in this paper is the minimization of the total power consumption of the proposed system while satisfying users data-rate requirements under blockage conditions. Simulation results demonstrate that the proposed framework significantly reduces the power consumption and enhances connectivity, reliability, sensing accuracy, and \ac{EE}, thereby improving overall system robustness compared to benchmark schemes and highlighting its suitability for practical blockage-prone indoor \ac{ISAC} environments.}

The major contributions are as follows: 

\begin{itemize}
     \item We propose a {reliable and energy-efficient \ac{ISAC} model} comprising a terahertz sensing access point \ac{$THz_s-AP$} and a hybrid terahertz/visible light communication access point \ac{$THz_c/VLC_c-AP$} within an \ac{ISAC} architecture. In the proposed framework, users detected by the \ac{$THz_s-AP$} are served by the \ac{$THz_c-AP$}, while the remaining users are served by the \ac{$VLC_c-APs$}, both in the presence and absence of blockages.
    \item We optimize the transmit power allocation between the \ac{$THz_s-AP$} and \ac{$THz_c-AP$} by imposing individual sensing and communication \ac{SNR} constraints. The \ac{$THz_s-AP$} power is dynamically adjusted to guarantee reliable sensing and detection accuracy, whereas the \ac{$THz_c-AP$} power is allocated to ensure that communication users achieve the minimum required \ac{SNR}.
    \item We formulate a mixed-integer linear programming (MILP) optimization to minimize the total network power consumption by optimizing the transmit powers of the \ac{$VLC_c-APs$}, where only active \ac{APs} consume power, ensuring energy-efficient operation in the presence and absence of human blockages.
    \item We evaluate the monostatic sensing \cite{bjornson2024introduction}, performance of the \ac{$THz_s-AP$} in terms of probability of detection (\ac{$P_d$}), probability of false alarm (\ac{$FA_p$}), and sensing coverage probability (\ac{$SC_p$}) under the impact of blockage and non-blockage scenario. 
    \item {We compare the proposed framework average \ac{SE} in the presence and absence of blockages with the standalone model: considered a model comprising a \ac{$THz_s-AP$} for user sensing and a \ac{$THz_c-AP$} for communication, where only the users detected by the \ac{$THz_s-AP$} are served by the \ac{$THz_c-AP$} only and the benchmark model: a heterogeneous LiFi/WiFi system proposed in \cite{karoti2022improving}, all \ac{APs} coordinated by a \ac{NC}}.
    \item We calculate the \ac{EE} of the proposed framework and compare it with that of a non-optimized baseline in the presence and absence of blockages.
    \item We investigate the trade-off between sensing and communication rates as the number of users increases, and also, analyze its variation under different blockage densities ($\lambda_B$) {and evaluate other key performance metrics, including user average \ac{EE}, average \ac{SE}, average communication rate, and average sensing rate, while explicitly considering the impact of human blockages.}
\end{itemize}

The remainder of this paper is organized as follows. Section II introduces the system model, including the network architecture, sensing model, blockage model, and channel models for both \ac{THz} and \ac{VLC} links. Section III presents the power allocation strategy \textcolor{blue}{between \ac{$THz_s-AP$} and \ac{$THz_c-AP$}}, while Section IV formulates the overall system power minimization problem. Section V presents the \ac{SNR} analysis of human-induced blockages. {Section VI presents the performance analysis of average \ac{SE}, and average \ac{EE}. Section VII discusses the simulation results, and Section VIII outlines the limitations of the proposed framework. Finally, Section IX concludes the paper with key insights and findings.}

\begin{table}[h]
\centering
\caption{Summary of Notations}
\label{tab:notations}
\renewcommand{\arraystretch}{1.2} 
\begin{tabular}{cl}
\hline
\textbf{Symbol} & \textbf{Description} \\ \hline
$N$        & Number of users \\
$VLC-AP$        & Visible light communication-access point \\
$VLC_c-AP$        & Visible light for communication-access point \\
$THz-AP$        & Terahertz communication-access point \\
$THz_s-AP$        & Terahertz for sensing-access point \\
$VLC_c-APs$        & Visible light for communication-access points \\
$THz_c-AP$        & Terahertz for communication-access point \\
$P_{\text{w}}$ & Transmit power for THz sensing and communication \\
$PSD$    & Power spectral noise density \\
$\rho_1$     & Power allocation factor (sensing/communication) \\
$\gamma_{sens}$ & THz sensing signal-to-noise ratio threshold \\
$\gamma_{comm}$ & THz communication signal-to-noise ratio threshold \\
$\gamma_{VLC}$ & VLC communication signal-to-noise ratio threshold\\
$SE$    & Spectral efficiency \\
$EE$    & Energy efficiency \\
$P_d$ & Probability of detection\\
$FA_p$ & Probability of false alarm\\
$P^{th}_d$ & Probability of detection threshold\\
$p_{max}$ & Maximum power\\
$w/o$ & Without \\
\hline
\end{tabular}
\end{table}
\begin{figure}
    \centering
    \includegraphics[width=1 \linewidth, height=0.6 \linewidth]{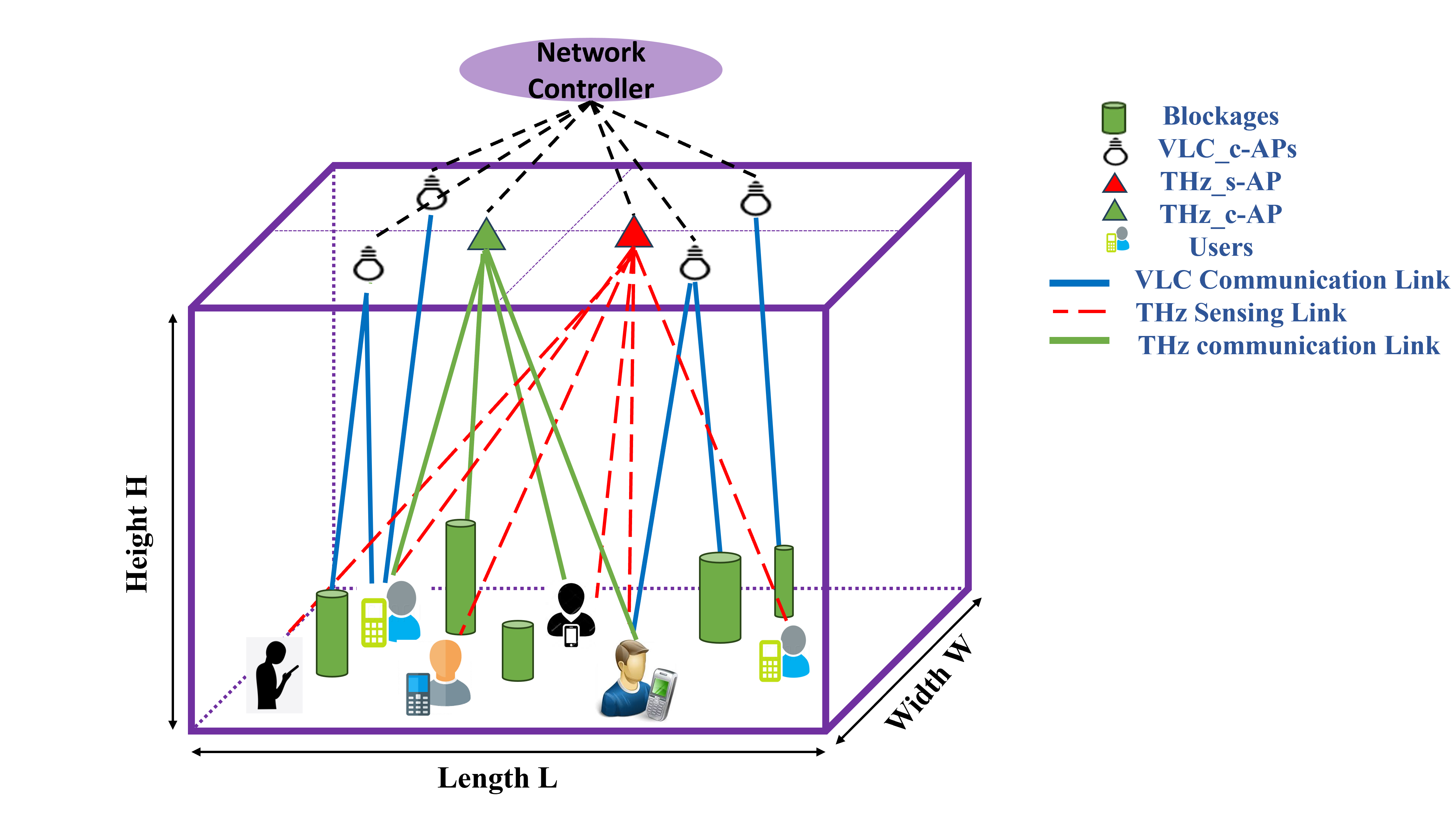}
    \caption{THz sensing with hybrid THz/VLC communication under the impact of human blockages.}
    \label{fig:placeholder}
\end{figure}

\vspace{-0.2cm}
\section{System Model}
We consider an indoor hybrid wireless communication architecture, as illustrated in Fig. 1, representing an enclosed space of dimensions \( H \times L \times W \)~meters (m). The setup consists of a heterogeneous deployment with a \ac{$THz_s-AP$}, \ac{$THz_c-AP$}, and \ac{$VLC_c-APs$}. All \ac{APs} are ceiling-mounted to ensure uniform spatial coverage and minimize obstruction effects. This deployment also provides robust \ac{LoS} connectivity, which is critical for the reliable operation of high-frequency \ac{THz} and \ac{VLC} links.

The monostatic sensing-based \ac{$THz_s-AP$} is responsible for user detection, where it identifies the presence of users based on a predefined sensing threshold $\gamma_{\text{sens}}$. Once sensing is performed, the \ac{$P_d$} is calculated, which then determines whether the user will be served by the \ac{$THz_c-AP$} or the \ac{$VLC_c-APs$} for communication. Users are assumed to be uniformly distributed in the indoor environment. Conversely, users who do not meet the sensing threshold are offloaded to the \ac{$VLC_c-APs$} for communication. This design enables a hybrid operation, where the \ac{$THz_s-AP$} supports intelligent user sensing, and the \ac{$THz_c/VLC_c-AP$} ensures efficient communication based on the users spatial positions and channel conditions. Both user association and link establishment are governed by the strict \ac{LoS} requirements, motivated by the directional nature of \ac{$THz_c/VLC_c-AP$} links. These links are highly sensitive to blockage and misalignment. Hence, this ensures \ac{LoS} guarantees high \ac{SNR} directional communication, resulting in robust and reliable connectivity even under dynamic blockage scenarios. To capture the impact of human presence on wireless propagation, human blockages are modeled using a \ac{MHCP-II}. This stochastic geometry model accounts for the spatial distribution and repulsion between blockers, preventing unrealistic overlap and enabling a more accurate representation of blockage induced variations in \ac{THz} and \ac{VLC} coverage. A \textcolor{blue}{\ac{NC}}, connected to all \ac{APs} via a high-capacity backhaul, manages network operations. The \textcolor{blue}{\ac{NC}} dynamically associates users to \ac{APs} and allocates resources by considering traffic demands, instantaneous channel conditions, and \ac{QoS} requirements. This ensures efficient load balancing, fairness among users, and enhances overall network performance. 

\subsection{Channel Model} 
\textit{Terahertz Communication Channel Model}: For \ac{THz}-band indoor communications, the \ac{LoS} channel gain between the \ac{$THz_c-AP$} and the $n^{th}$ user is denoted by $H_{n}^{\textnormal{THz}_{c}}$. In contrast to microwave and mmWave frequencies, \ac{THz} propagation is strongly influenced by two dominant factors: (i) free-space spreading loss, caused by geometric signal attenuation with distance, and (ii) molecular absorption loss, arising from frequency-selective absorption of \ac{THz} signals by atmospheric molecules (e.g., water vapor and oxygen). The free-space spreading component is modeled as \cite{prabhakar2024thz}: 
\begin{equation}
H_\textnormal{spr} = \frac{c}{4 \pi d f},
\end{equation}
where $d$ is the distance between the \ac{$THz_c-AP$} and the user, $c$ is the speed of light, and $f$ denotes the carrier frequency. The molecular absorption effect is represented by \cite{prabhakar2024thz}:
\begin{equation}
H_\textnormal{abs} = \exp\!\left(-\tfrac{1}{2} k(f) d \right),
\end{equation}
with $k(f)$ denoting the frequency-dependent molecular absorption coefficient. 

Accordingly, the overall \ac{LoS} channel gain for the \textcolor{blue}{$n^{th}$} user is expressed \cite{prabhakar2024thz}:
\begin{equation} \label{Eq:1}
H_{n}^{\textnormal{THz}_{c}} = H_\textnormal{spr}\, H_\textnormal{abs}.
\end{equation}

\textit{Visible Light Communication Channel Model}: In this paper, we adopt the \ac{LoS}  \ac{VLC} channel model for the downlink between a \ac{$VLC_c-APs$} and the users. The \ac{LoS} channel gain at \textcolor{blue}{$n^{th}$} user is denoted by $\textit{H}^{\text{VLC}_{c}}_{{n}}$ and expressed as \cite{wu2017access}:
\begin{equation}
\textit{H}^{\text{VLC}_{c}}_{{n}} =
\begin{cases}
\frac{(m+1)A}{2\pi D^2}\cos^m(\rho)\, T_{s}(\psi)\, g(\psi)\, \cos(\theta), & 0 \leq \psi \leq \psi_c, \\[2ex]
0, & \psi > \psi_c,
\end{cases}
\label{eq212}
\end{equation}
where $A$ is the area of the photodetector (PD) of the user receiver, $\theta$ is the angle of incidence at the receiver, $\rho$ is the irradiance angle at the \ac{$VLC_c-APs$} transmitter. $D$ is the distance between the \ac{$VLC_c-APs$} and the \textcolor{blue}{$n^{th}$} user. The term $T_{s}(\psi)$ denotes the optical filter gain, while $g(\psi)$ represents the optical concentrator gain.  

The Lambertian order $m$ of the \ac{$VLC_c-APs$} radiation pattern is given by \cite{wu2017access}:
\begin{equation}
m = \frac{-\ln(2)}{\ln\!\left(\cos(\textcolor{blue}{\phi_n})\right)},
\label{phiangle}
\end{equation}
where \textcolor{blue}{$\phi_n$} denotes the \ac{$VLC_c-APs$} semi-angle at half illuminance.  

The concentrator gain $g(\psi)$ is defined as\cite{wu2017access}:
\begin{equation}
g(\psi) =
\begin{cases}
\dfrac{ci^2}{\sin^2(\psi_c)}, & 0 \leq \psi \leq \psi_c, \\[2ex]
0, & \psi > \psi_c,
\end{cases}
\end{equation}
where $ci$ is the refractive index of the optical concentrator and $\psi_c$ represents the field of view (FOV) of the receiver. 

{ The illumination provided by the VLC access point is expressed using a commonly adopted formulation as \cite{khreishah2018hybrid} :
\begin{equation}
\Phi_{(VLC_c-APs)} = 683 \int_{0}^{\infty} V(\lambda)\, P(\lambda)\, d\lambda,
\end{equation}
where $V(\lambda)$ denotes the standard luminous efficiency function~\cite{ghassemlooy2019optical}.
Here, $P(\lambda)$ denotes the spectral power distribution of the \ac{$VLC_c-APs$}, which depends on both the average transmitted optical power and the \ac{$VLC_c-APs$}  type. It satisfies
\begin{equation}
\int_{0}^{\infty} P(\lambda)\, d\lambda =  p_{\max}.
\end{equation}}

{Accordingly, the luminous flux can be expressed as:
\begin{equation}
\Phi_{(VLC_c-APs)} ~[\mathrm{lm}] = G_m  p_{\max},
\end{equation}
where $G_m$ is a constant determined by the characteristics of the \ac{$VLC_c-APs$}. The total illumination contributed only by the active \ac{$VLC_c-APs$} at a given location of $n^{th}$ users at a distance $D$ can therefore be written as: 
\begin{equation}
\phi_n~[\mathrm{lm}] = \sum \textit{H}^{\text{VLC}_{c}}_{{n}} \Phi_{VLC_c-APs}. 
\end{equation}
Here, $\Phi_{VLC_c-APs}$ is a linear function of $\phi_n$, since all other parameters are typically constant. In the illumination model, a minimum horizontal illuminance constraint is imposed on each $\phi_n$ (e.g., 300~lux for a typical office environment). The contribution of ambient light can be incorporated by adding an additional term to the linear expression of $\phi_n$.}

\textit{Terahertz Monostatic Sensing Model}: THz communication channels model one-way links focusing on maximizing data rate, accounting for path loss, blockages, and noise. In contrast, THz sensing channels model round-trip reflection from targets, incorporating target radar cross-section (RCS) and higher path loss ($d_{f}^{-4}$), with \ac{SNR} used for detection rather than data transmission. For user sensing, we use the \ac{$THz_s-AP$} based path loss model, where the path loss function for the $n^{th}$ user is denoted as $L(\textcolor{blue}{d_{f}})_{\textcolor{blue}{(n)}}$ and is expressed as \cite{10622748}:
\begin{equation}
L(\textcolor{blue}{d_{f}})_{\textcolor{blue}{(n)}} = \frac{c^2 \, \sigma_{\text{RCS}}}{(4\pi)^3 f^2 d_{f}^4} \exp\!\left(-2k(f)d_{f}\right),
 \label{eq:7}
\end{equation}
where $d_{f}$ is the distance between the \ac{$THz_s-AP$} and the \textcolor{blue}{$n^{th}$} user. Also the $k(f)$ represents the frequency-dependent molecular absorption coefficient. 

The RCS of the user, denoted by $\sigma_{\text{\textcolor{blue}{RCS}}}$, is modeled as an exponentially distributed random variable \cite{10622748}:
\begin{equation}
\sigma_{\text{\textcolor{blue}{RCS}}} \sim \text{Exp}(1).
\end{equation}

\textit{Human Blockage Analysis Model}: In indoor environments, human blockages are modeled as cylindrical obstacles that obstruct the \ac{LoS} path between an \ac{APs} and users, thereby degrading link quality. To capture the spatial distribution of such blockages in a realistic manner, the \ac{MHCP-II} is employed \cite{singh2021performance},\cite{matern2013spatial}. The \ac{$THz_s-AP$} and {$THz_c/VLC_c-APs$} are assumed to be mounted at a ceiling height $H$, while human blockages of height $h_\textnormal{B}$ are distributed across the user ground plane. The horizontal separation between the \ac{APs} and a user is denoted by $d_\textnormal{T}$. Shown in Fig. 2, a blockage positioned at a distance $d_\textnormal{B}$ from the \ac{APs} is therefore located $d_\textnormal{T} - d_\textnormal{B}$ away from the user. 
\begin{figure}
    \centering    \includegraphics[width=0.80\linewidth, height=0.5 \linewidth]{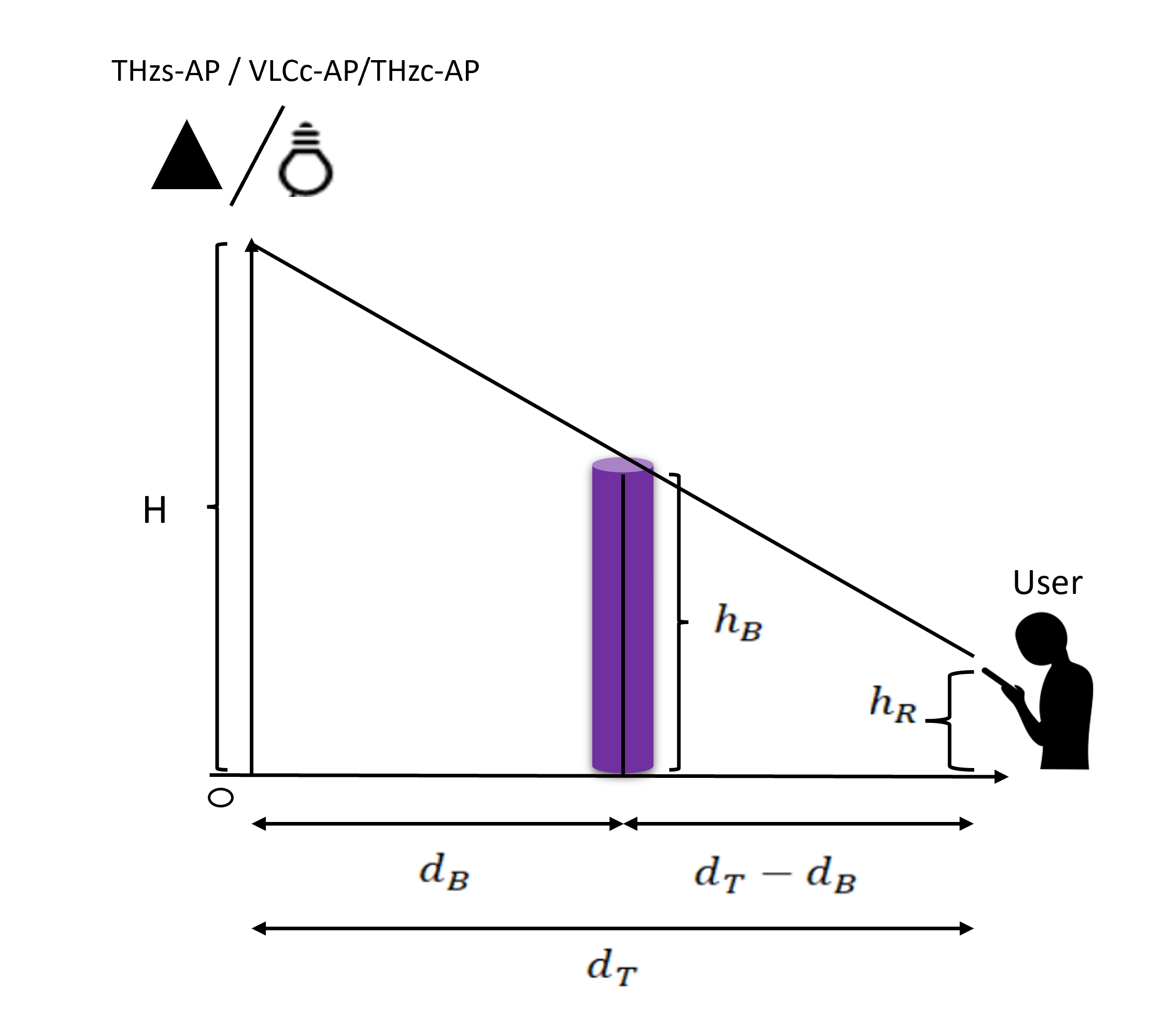}
    \caption{Schematic diagram for calculation of blockage distances.}
    \label{fig:placeholder}
\end{figure}
The geometric blocking condition can be expressed as \cite{singh2021performance}:
\begin{equation}
d_\textnormal{T} - d_\textnormal{B} = \frac{h_\textnormal{B}}{H}\, d_\textnormal{T}.
\end{equation}

The effective blockage density ($\lambda_\textnormal{B}$) under the MHCP-II model, with hard-core distance $\delta$, is given by \cite{singh2021performance}:
\begin{equation}
\lambda_\textnormal{B} = \frac{1-\exp(\lambda_\textnormal{p}\pi \delta^2)}{\pi \delta^2},
\end{equation}
where $\lambda_\textnormal{p}$ denotes the intensity of the baseline {Poisson Point Process (PPP)}.  

The probability that a user located at distance $d_\textnormal{B}$ from a potential blocker experiences \ac{LoS} obstruction is expressed as \cite{singh2021performance}:
\begin{equation} \label{Eq:14}
P_{B}(d_\textnormal{B}) = 1 -   \exp\!\left(-2\lambda_\textnormal{B} d_\textnormal{B} r_\textnormal{B}^2\right),
\end{equation}
where $r_\textnormal{B}$ is the radius of the cylindrical blockage.  
\section{{Power allocation between THz sensing and communication }}
In a joint sensing and communication framework, the total transmit power \( P_{\text{w}} \) is allocated between the two functionalities through a power-splitting factor \( \rho_1 \in [0,1] \). Specifically, a fraction \( \rho_1 \) of the power is assigned to sensing, while the remaining fraction \( (1-\rho_1) \) is devoted to communication, and {$P_{\mathrm{THz(cir)}}$ \cite{nguyen2025energy} is the circuitry power.} The achievable \ac{SNR} for each functionality depends on the corresponding \ac{SNR}, defined as \(\frac{L(\textcolor{blue}{d_{f}})_{\textcolor{blue}{(n)}}}{N_{\text{THz}_{s}} \, BW_{\text{THz}_{s}}}\) for the  sensing user and \(\frac{H_{n}^{\text{THz}_{c}}}{N_{\text{THz}_{c}} \, BW_{\text{THz}_{c}}}\) for the communication user. To guarantee reliable performance, both sensing and communication must satisfy predefined \ac{SNR} thresholds, represented by \( \gamma_{\text{sens}} \) and \( \gamma_{\text{comm}} \), respectively. 
\begin{equation}
(P_{\text{w}}+\textcolor{blue}{P_{\mathrm{THz(cir)}}}) \, (\frac{H_{n}^{\text{THz}_{c}}}{N_{\text{THz}_{c}} \, BW_{\text{THz}_{c}}}) \, (1-\rho_1) \;\geq\; \gamma_{\text{comm}}, \label{eq:12}
\end{equation}
\begin{equation}
(P_{\text{w}}+\textcolor{blue}{P_{\mathrm{THz(cir)}}}) (\frac{L(\textcolor{blue}{d_{f}})_{\textcolor{blue}{(n)}}}{N_{\text{THz}_{s}} \, BW_{\text{THz}_{s}}}) \, (\rho_1) \;\geq\; \gamma_{\text{sens}},\label{eq:13}
\end{equation}
\begin{equation}
0\leq\rho_{1} \leq 1.
\end{equation}

The constraint in \eqref{eq:12} ensures that the achieved \ac{SNR} of the \ac{$THz_c$} link meets or exceeds the threshold $\gamma_{\text{comm}}$, thereby guaranteeing reliable communication performance. Similarly, the sensing \ac{SNR} constraint in \eqref{eq:13} corresponds to the \ac{$THz_s$} link in a based joint sensing and communication system. It ensures that the power allocated for sensing is sufficient to achieve the required \ac{SNR} for accurate user detection. ${N_{\text{THz}_{s}}}$ and ${BW_{\text{THz}_{s}}}$ represent the noise \ac{PSD} and bandwidth of the \ac{THz} sensing, while ${N_{\text{THz}_{c}}}$ and ${BW_{\text{THz}_{c}}}$ denote the noise \ac{PSD} and bandwidth of the \ac{THz} communication.
The inequalities ensure that the allocated transmit power is sufficient to meet the 
minimum \ac{QoS} requirements for both sensing and communication. Consequently, 
the optimization of \( \rho_1 \) plays a critical role in balancing the trade-off between 
the two functionalities, directly impacting the feasibility and efficiency of the \ac{ISAC} systems.

\section{Total system Power minimization}
In this section, we analyze the total power consumption of our proposed model. We adopt a user-centric strategy in which only the \ac{$VLC_c-APs$} associated with active users remain operational, while all non-associated \ac{$VLC_c-APs$} are switched off. The total transmit power is optimized such that each user is served exclusively by its associated \ac{AP}, and all other \ac{APs} remain inactive. This adaptive activation mechanism ensures that unnecessary power consumption from idle \ac{APs} is eliminated, thereby significantly reducing overall system power usage and enhancing the \ac{EE} of the network. The total power consumption is defined as the power required to establish and maintain reliable communication links between the \ac{$VLC_c-APs$} and users through their respective transmission processes. 

We consider a system comprising a \ac{$VLC_c-APs$}, denoted by 
$\mathcal{L} = \{1, \dots, L\}$, and \textcolor{blue}{each $n^{th}$} user is associated 
with a serving \ac{AP} $\ell \in \mathcal{L}$. The binary activity indicator of each AP, 
$\alpha_\ell \in \{0,1\}, \ \forall \ell \in \mathcal{L}$, and  $\rho_{l}$
representing the total transmit power allocation of each AP, $\rho_\ell \geq 0, \ \forall \ell \in \mathcal{L}$. 

The optimization problem can be cast as:
\begin{subequations} \label{eq:14}
\begin{align}
& \operatorname*{minimise}_{\rho_{l}, \alpha_{l} \; \forall l} 
&& \sum_{l=1}^{L} \alpha_{l}\rho_{l} +(P_{\text{w}}+\textcolor{blue}{P_{\mathrm{THz(cir)}}}),
\label{eq:14a} \\[6pt]
& \text{Subject to} 
&& H^{\text{VLC}_{c}}_{n} \, \rho_{l} \;\;\geq\;\; \gamma_{\text{VLC}} \, N_{VLC_{c}}BW_{VLC_c},
\notag \\[-3pt]
&&& \forall \, l \in \{1,\dots,L\},
\label{eq:14b} \\[6pt]
&&& \rho_{l} \leq \alpha_l \, p_{\max}, 
\quad \forall \, l \in \{1,\dots,L\},
\label{eq:14c} 
\end{align}
\end{subequations}

{here, $\alpha_l$ is a binary activation variable defined as:
\[
\alpha_{l} =
\begin{cases}
1, & \text{if the $l^{th}$ VLC AP is active}, \\
0, & \text{otherwise}.
\end{cases}
\]}

The constraints (\ref{eq:14b}) satisfy the minimum \ac{SNR} requirement for each $n^{th}$ user, where $\gamma_{VLC}$, ${N_{\text{VLC}_{c}}}$, and ${BW_{\text{VLC}_{c}}}$ denote as the \ac{$VLC_c-APs$} \ac{SNR} threshold, noise \ac{PSD}, and bandwidth of the \ac{VLC} respectively. The constraint in (\ref{eq:14c}) provides the transmit power constraint per \ac{$VLC_c-AP$}. The binary activation variable $\alpha_l$ indicates whether the $l^{th}$ \ac{$VLC_c-AP$} is active, where $\alpha_l = 1$ denotes that the AP is switched on and $\alpha_l = 0$ denotes that it is switched off. {The variable $\rho_l$ represents the transmit power allocated when the $n^{th}$ user is associated with the $l^{th}$ \ac{VLC-AP}. In this work, for the \ac{$VLC_c-APs$}, each $n^{th}$ user is associated with only one VLC-AP, selected to provide the maximum \ac{SNR}}

This is a non-convex MILP due to the presence of the binary activation variable $\alpha_l$. Using this optimization, the power minimization problem with binary AP selection is solved through the CVX-MOSEK solver \cite{cvx} and \cite{mosek}.

\section{SNR Analysis of THz Sensing with Hybrid THz/VLC Communication under Impact of Human Blockages} 

{
For each channel model corresponding to \ac{$THz_s-AP$}, \ac{$THz_c-AP$}, and \ac{$VLC_c-APs$}, the received signal is modeled as a function of the user distances, denoted by $d$, $D$, and $d_{f}$, respectively. Based on these distances, the proposed \ac{SNR} expressions are derived for each link.}

{To incorporate the impact of blockage, the computed \ac{SNR} values are weighted by the blockage probability that is $(P_{B}(d_{B}))$, which depends on the blockage parameters $\lambda_B$, $d_B$, and $r_B$. The simulation is performed over 1000 Monte Carlo iterations for each user, and the resulting \ac{SNR} values are averaged over all iterations.}

{The \ac{SNR} performance of the proposed system under blockage conditions is first evaluated for \ac{THz} sensing, followed by hybrid \ac{THz}/\ac{VLC} communications. Specifically, if the $n^{th}$ user lies within the sensing coverage and achieves a probability of detection (\ac{$P_d$}) under the impact of blockages denoted as $P_{d(n)}^{(B)}$ greater than the probability of detection threshold denoted as $P^{th}_d$, the user is associated with the \ac{$THz_c-AP$}. Otherwise, if $P_{d(n)}^{(B)} < P^{th}_d $, the user is served by the \ac{$VLC_c-APs$}. Among the \ac{$VLC_c-APs$}, the $n^{th}$ user is associated with that \ac{$VLC_c-AP$} which is closest to the user and provides the highest \ac{SNR}.}

\subsection{SNR analysis of \ac{$THz_s-AP$} and \ac{$THz_c-AP$}}
The sensing \ac{SNR} at the \ac{$THz_s-AP$}, considering the impact of human blockages, is denoted by $SNR^{\textcolor{blue}{\text{THz}_\text{s}}{(B)}}_{{n}} $ and is given as:
\begin{equation}
SNR^{\textcolor{blue}{\text{THz}_\text{s}}{(B)}}_{{n}} = \frac{G_t G_r L(\textcolor{blue}{d_{f}})_{\textcolor{blue}{(n)}}}{N_{\text{THz}_{s}} \, BW_{\text{THz}_{s}}} \cdot P_{B}(d_B),
\end{equation} where the $G_t$ and $G_r$ are the transmitter and receiver antenna gains of the \ac{$THz_s-AP$} and receiver of the user.  \( N_{\text{THz}_{s}} \) denotes the noise \ac{PSD} in the THz sensing channel, and \( BW_{\text{THz}_{s}} \) corresponds to the allocated bandwidth of the \ac{$THz_s-AP$}. The term $P_{B}(d_B)$ represents the probability of blockage at distance $d_B$, which models the attenuation caused by the human blockers.  

Similarly, for the \ac{THz} communication \ac{SNR} with the inclusion of human blockages, the {$THz_c-AP$} channel gain is expressed as:
\begin{equation}
H^{(B)}_{\text{THz}_{c}} = G_{t} G_{r} H_n^{\text{THz}_{c}} \cdot P_{B}(d_B),
\end{equation}
where, \( G_t \) and \( G_r \)
denote the transmit and receive antenna gains in the \ac{$THz_c-AP$} user’s receiver, respectively. $H_n^{\text{THz}_{c}}$ is defined in (\ref{Eq:1}). Therefore, the \ac{SNR} experienced by the $n^{th}$ user in the \ac{THz} channel in the presence of human blockages is defined as: 
\begin{equation}
\text{SNR}_{n}^{\mathrm{THz}_{c}(B)} = \frac{H^{(B)}_{\text{THz}_{c}}}{N_{\text{THz}_{c}} BW_{\text{THz}_{c}}}
\label{eq:SNR_WiFi},
\end{equation}
where the \ac{$THz_c-AP$}, \( N_{\text{THz}_{c}} \) denotes the noise \ac{PSD} in the THz channel, and \( BW_{\text{THz}_{c}} \) corresponds to the allocated bandwidth of the \ac{$THz_c-AP$}.

\subsection{Probability of detection}
A widely used decision strategy is the 
Neyman--Pearson criterion \cite{4815550}, which is a special case of the Bayes criterion. 
Under this approach, the objective is to maximize the 
\ac{$P_d$} while ensuring that the probability of false alarm (\ac{$FA_p$}) does not exceed a predefined threshold. This criterion is particularly important because, in practical sensing systems, false 
alarms can lead to wasted resources, misinterpretation of data, or unnecessary responses. By constraining \ac{$FA_p$}, the Neyman--Pearson framework provides a controlled trade-off between sensitivity (detecting true user) and reliability (avoiding false user). The $P_d$ serves as a key metric for sensing performance. It represents the likelihood that the system correctly identifies the presence of a target when the target is indeed present\cite{4815550}. 
\begin{lemma}
The \ac{$P_d$} of the \textcolor{blue}{$n^{th}$ user} in the proposed network is under the impact of human blockages, and is denoted by $P_{d(n)}^{(B)}$, and calculated as follows:
\begin{equation}
\begin{aligned}
P_{d(n)}^{(B)} &= \frac{1}{2} \left[ 1 - \operatorname{erf} \left( \operatorname{erf}^{-1}(1 - 2\ac{$FA_p$}) - \sqrt{SNR^{\textcolor{blue}{\mathrm{THz}_\text{s}}{(B)}}_{{n}}/2 } \right) \right] 
\end{aligned}
\end{equation}
\begin{equation}
\begin{aligned}
P_{d(n)}^{(B)}&= \frac{1}{2} \operatorname{erfc} \left( \operatorname{erf}^{-1}(2\ac{$FA_p$}) - \sqrt{SNR^{\textcolor{blue}{\mathrm{THz}_\text{s}}{(B)}}_{{n}}} {/2} \right)
\end{aligned}
\end{equation}
 
\textit{Proof}: The \textcolor{blue}{$P_{d(n)}^{(B)}$} depends on the $SNR^{\textcolor{blue}{\mathrm{THz}_{s}}{(B)}}_{{n}} $ and the \ac{$FA_p$}. Users within the coverage area are not always detected correctly due to false alarms. Hence, \ac{$P_d$} plays a crucial role in ensuring the accuracy of user presence detection within the sensing coverage.
\end{lemma}
\subsection{SNR analysis of \ac{$VLC_c-APs$}}
Similarly, for the \ac{$VLC_c-APs$} communication model, the channel gain in the presence of blockages is given by:
\begin{equation}   
H^{(B)}_{\text{VLC}_{c}} = H_n^{\text{VLC}_{c}} \cdot P_{B}(d_B).
\end{equation}

Further, the \ac{SNR} experienced by the $n^{th}$ user from the \ac{$VLC_c-APs$} channel under the blockage impacts is described as: 
\begin{equation}
SNR_n^{\mathrm{VLC}_{c}(B)} = \frac{\left( R H^{(B)}_{\mathrm{VLC}_{c}} / k \right)^2}{N_{\mathrm{VLC_{c}}} \cdot BW_{\mathrm{VLC_{c}}}},
\label{eq:SNR_VLC}
\end{equation}
where \( N_{\text{VLC}_{c}} \) denotes the noise PSD in the \ac{$VLC_c-APs$} channel, and \( BW_{\text{VLC}_{c}} \) corresponds to the available bandwidth of the \ac{$VLC_c-APs$}. Additionally, \( R \) denotes the responsivity of the receiver, while \( k \) is the coefficient for optical-to-electrical power conversion.
\begin{lemma}
The maximum received \ac{SNR} experienced by $n^{th}$ user from \ac{$VLC_c-APs$} system under the impact of blockages is denoted by $\text{SNR}_n^{\rm{max}\textit{(B)}}$ and expressed as: 
\begin{equation}
\text{SNR}_n^{\rm{max}\textit{(B)}}=  \underset{}{\operatorname*{max}} (SNR_n^{\mathrm{VLC}_{c}(B)}).
\end{equation}

Proof: We consider $\text{SNR}_n^{\rm{max}\textit{(B)}}$, which selects from the \ac{$VLC_c-APs$} providing the highest \ac{SNR}. Since \ac{SNR} generally decreases with distance, this association strategy effectively corresponds to selecting the closest \ac{$VLC_c-AP$} for each user.
\end{lemma}

\section{{Performance Analysis}}
{In this section, we evaluate the overall system performance in terms of average \ac{SE} and average \ac{EE}. The analysis highlights the effectiveness of the proposed framework in maintaining reliable communication performance under human blockage conditions while minimizing power consumption and enhancing overall system robustness.}
\subsection{Average \ac{SE}}
\begin{lemma}
The average \ac{SE} of the proposed network, denoted by $S_{H}^{\rm avg}$, is computed based on the probability of detection threshold {($P_d^{\rm th}=0.5$), which determines whether the $n^{th}$ user is associated with the \ac{$THz_c-AP$} or the \ac{$VLC_c-APs$}}, as follows:

\begin{equation}
S_{H}^{\mathrm{THz}_{c}/\mathrm{VLC}_{c}} =
\begin{cases}
\text{SNR}_{n}^{\mathrm{THz}_{c}(B)}, & P_{d(n)}^{(B)} > \textcolor{blue}{P^{th}_d}, \\[10pt]
\text{SNR}_{n}^{\textnormal{max}{(B)}},    & P_{d(n)}^{(B)} \leq \textcolor{blue}{P^{th}_d},
\end{cases}
\end{equation}
\begin{equation}
S_{H}^{\mathrm{avg}} = 
\frac{1}{N} \sum_{n=1}^{N_{AP}} 
\log_2 \Bigl( 1 + \text{S}_{H}^{\mathrm{THz_{c}/VLC_{c}}} \Bigr),
\end{equation}
where $N_{AP}$ is the number of users associated with either \ac{$THz_c-AP$} or \ac{$VLC_c-AP$} and $N$ is the number of users.

\textit{Proof}: The average \ac{SE} of the proposed model is evaluated after user sensing by the \ac{$THz_s-AP$}. The $P^{th}_{d}$ determines whether a user will be served by the \ac{$THz_c-AP$} or the \ac{$VLC_c-AP$}. We further analyze the performance under the impact of blockages, where sensing performance parameters are significantly impacted. In such a case, most users can establish a reliable connection through the \ac{$VLC_c-APs$}, while only a few are able to maintain service through the \ac{$THz_c-AP$}.
\end{lemma}
\subsection{Average \ac{EE}}
\begin{lemma}
The average \ac{EE} of the proposed network, denoted by $\eta^{\rm{THz/VLC}}$, is obtained by dividing the \ac{SE} by the total power consumption, from the \ac{$THz_s-AP$}, \ac{$THz_c-AP$} and \ac{$VLC_c-APs$} is expressed as follows:
\begin{equation}
\eta^{\mathrm{THz/VLC}} = \frac{S_{H}^{\mathrm{avg}}}{(P_{\mathrm{THz_s}} + P_{\mathrm{THz_c}}+\textcolor{blue}{P_{\mathrm{THz(cir)}}})+\alpha P_{\mathrm{VLC_c}}},
\end{equation}
where $P_{\mathrm{THz_s}}$, $P_{\mathrm{THz_c}}$, and $P_{\mathrm{VLC_c}}$ denote the transmit powers of the \ac{$THz_s-AP$}, \ac{$THz_c-AP$}, and \ac{$VLC_c-APs$}, respectively, and \textcolor{blue}{$P_{\mathrm{THz(cir)}}$ \cite{nguyen2025energy} represents the circuitry power consumption of the \ac{$THz_c$} and \ac{$THz_s$}.}

Proof: The average \ac{EE} of the proposed model accounts for the power consumption of \ac{$THz_s-AP$}, \ac{$THz_c-AP$}, and \ac{$VLC_c-APs$} components. A higher value of $\eta^{\mathrm{THz/VLC}}$ indicates that the system achieves greater \ac{SE} per unit of power, reflecting enhanced energy-efficient performance under the hybrid sensing and communication framework. Note that the total power of \ac{$VLC_c-APs$} is multiplied by $\alpha$, corresponding to the deployed number of \ac{$VLC_c-APs$}.
\end{lemma} 
\begin{table}[!t]
\centering
\caption{Simulation Parameters \cite{singh2021performance},\cite{prabhakar2024thz},\cite{karoti2022improving}}
\begin{tabular}{|l|c|}
\hline
\textbf{Parameter (Symbol)} & \textbf{Value}\\
\hline
\multicolumn{2}{|c|}{\textbf{THz}} \\
\hline
Central carrier frequency ($f_c$) & 370 GHz \\
THz sensing and communication power ($P_{\text{w}}$) & 2 W \\ 
\textcolor{blue}{THz sensing and communication circuitry power ($P_{\mathrm{THz(cir)}}$)} & \textcolor{blue}{5.6 mW} \\ 
Bandwidth of $THz_s-AP$ and $THz_c-AP$($B_{\text{THz}}$) & 100 MHz \\
THz Noise PSD ($N_{\text{THz}}$) & $-174$ dBm/Hz \\
THz Sensing threshold  \textcolor{blue}{($\gamma_{sens}$)} & -5 dB \\
Probability of false alarm ($FA_{p}$) & $10^{-2}$ \\
\textcolor{blue}{Probability of detection threshold ($P_{d}^{th}$)} & \textcolor{blue}{$0.5$} \\
THz communication threshold  ($\gamma_{comm}$) & 25 dB \\
Height of blockage ($H_B$) & 1.8 m \\
Radius of blockage ($r_B$) & 2.0 m \\
Density of blockages ($\lambda_B$) &  4 \\
THz antenna gains ($G_t$, $G_r$) & 1  \\
 \hline
\multicolumn{2}{|c|}{\textbf{VLC}} \\
\hline
Optical filter gain ($g_{\text{filter}}$) & 1 \\
VLC communication threshold  ($\gamma_{VLC}$) & 15 dB \\
Half-intensity radiation angle ($\varphi_{1/2}$) & 60$^\circ$ \\
PD Field of view (FOV) ($\theta_{\max}$) & 90$^\circ$ \\
Optical-to-electrical conversion efficiency ($k$) & 3 \\
Responsivity of receiver ($R$) & 0.53 A/W \\
Maximum power ($p_{max}$) & 5W \\
Noise PSD ($N_{\text{VLC}}$) & $-210$ dBm/MHz \\
Bandwidth of $VLC_c-APs$ ($B_{\text{VLC}}$) & 40 MHz \\
\hline
\end{tabular}
\end{table}
\section{Results and Discussions}
This section provides a rigorous validation of the analytical framework through {Monte Carlo} simulations, ensuring consistency with empirical evaluations.  The proposed system is evaluated under human blockage conditions in a \(5\,\text{m} \times 5\,\text{m} \times 3\,\text{m}\) indoor environment, {with  
\ac{$THz_c-AP$} located at $(3.0,\, 2.5,\, 2.8)$ and \ac{$THz_s-AP$} at $(1.5,\, 2.5,\, 2.8)$ respectively. Furthermore, \ac{$VLC_c-APs$} are deployed on the ceiling coordinates are 
$(1.25,\,1.25,\,2.8)$, $(1.25,\,3.75,\,2.8)$, $(3.75,\,3.75,\,2.8)$, and $(3.75,\,1.25,\,2.8)$ \cite{karoti2022improving}}. A number of users ($N$=10) were randomly positioned at a height of 0.85\,m above the ground plane. The corresponding system parameters, including blockage density ($\lambda_{B}$) and blockage radius, are summarized in Table II. This simulation setup enables a realistic characterization of signal attenuation and spatial coverage in blockage-prone scenarios. The system performance is evaluated with respect to average \ac{SE}, \ac{$P_d$}, \ac{$SC_p$}, $\rho_{1}$, and the average communication and sensing rates, all analyzed as functions of the user density within the environment. Each metric is averaged over 1000 realizations to ensure statistical reliability.
\begin{figure*}[htbp]
  \centering
  \begin{subfigure}[t]{0.24\textwidth}
    \centering    \includegraphics[width=1.1\linewidth,height=1\linewidth]{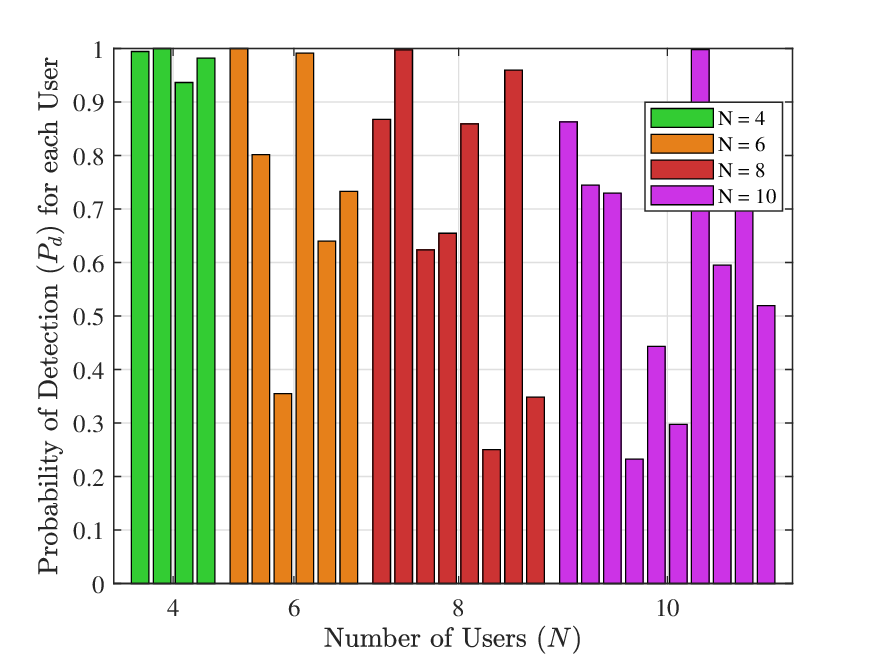}
    \caption{$P_d$ Vs. $N$ $w/o$ blockages.}
    \label{fig:3a}
  \end{subfigure}\hfill
  \begin{subfigure}[t]{0.24\textwidth}
    \centering    \includegraphics[width=1.1\linewidth,height=1\linewidth]{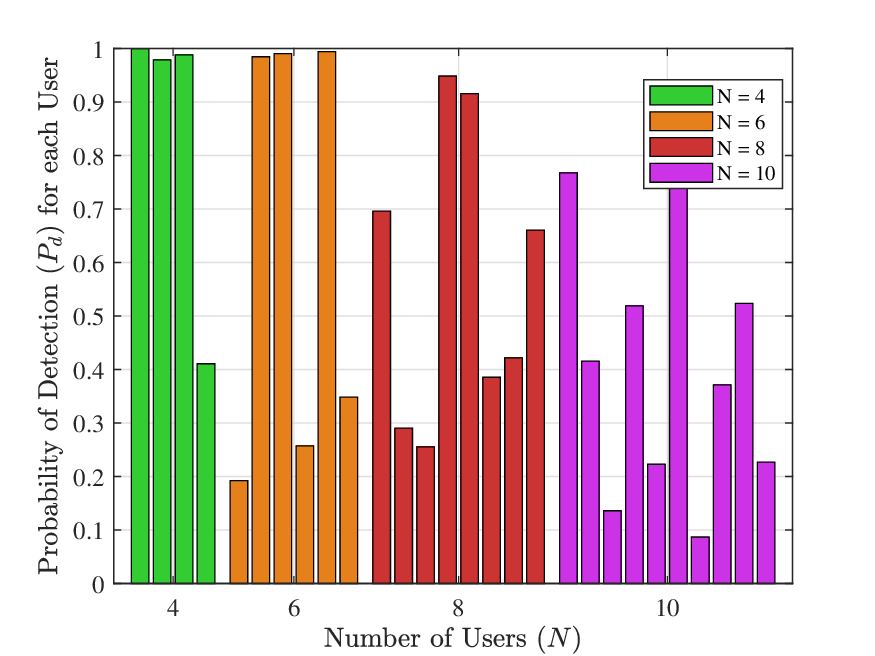}
    \caption{$P_d$ Vs. $N$ with blockages.}
    \label{fig:3b}
  \end{subfigure}\hfill
   \begin{subfigure}[t]{0.24\textwidth}
    \centering    \includegraphics[width=1.1\linewidth,height=1\linewidth]{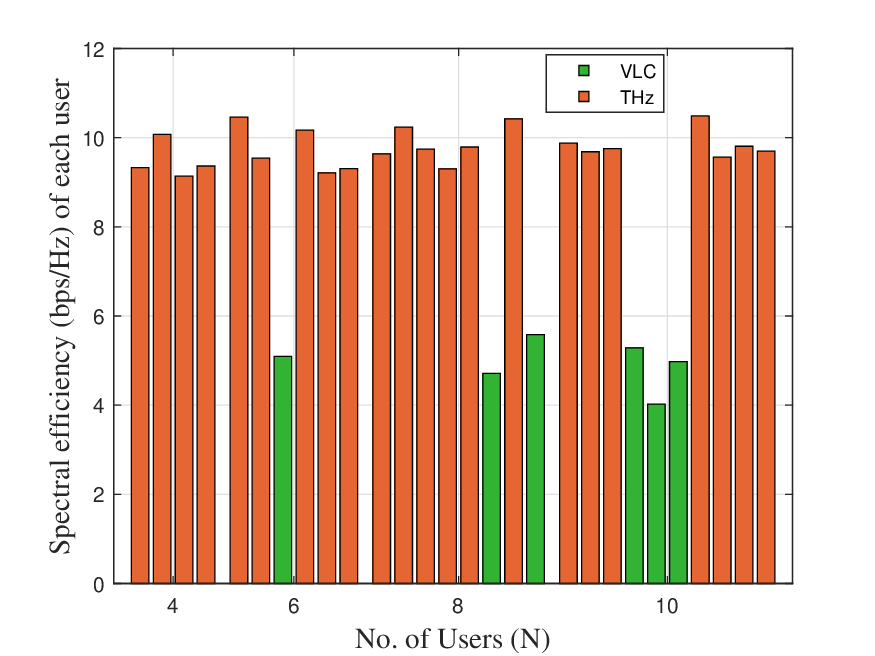}
    \caption{$SE$ Vs. $N$ $w/o$ blockages.}
    \label{fig:3c}
  \end{subfigure}\hfill
  \begin{subfigure}[t]{0.24\textwidth}
    \centering    \includegraphics[width=1.1\linewidth,height=1\linewidth]{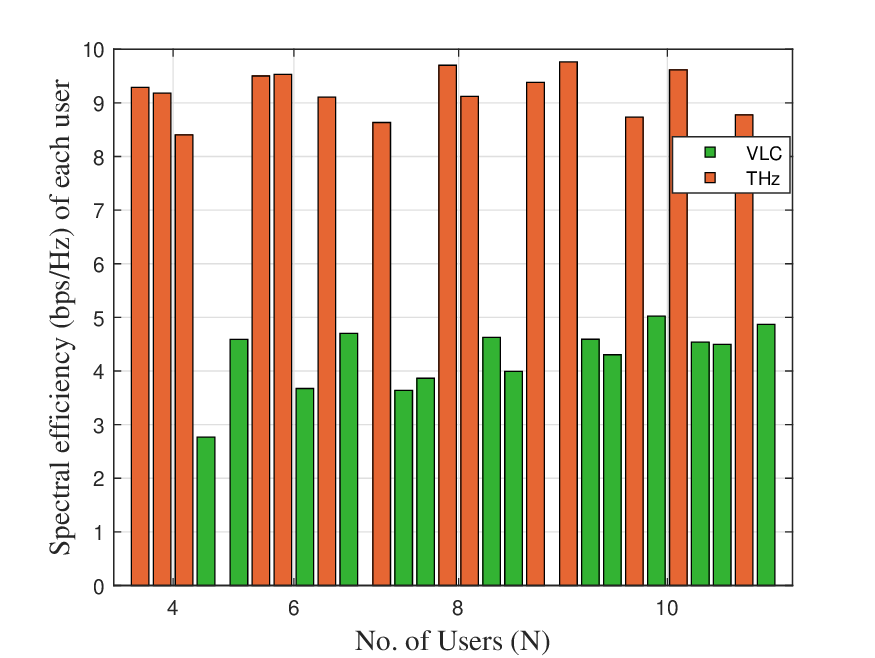}
    \caption{$SE$ Vs. $N$ with blockages.}
    \label{fig:3d}
  \end{subfigure} 
  \caption{\textcolor{blue}{Comparison analysis with probability of detection and spectral efficiency of each user in the presence and absence of human blockages.}}
  \label{fig:all}
\end{figure*}
\begin{figure}[htbp]
   \begin{minipage}[t]{0.45\textwidth}
    \centering  \includegraphics[width=1.05\linewidth,height=0.78\linewidth]{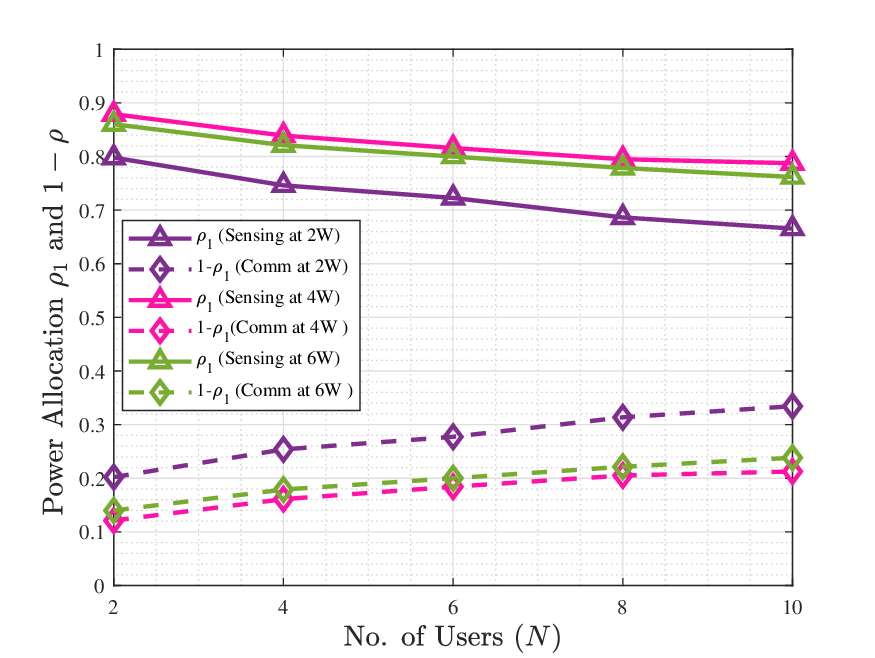}
    \caption{Power allocation between \ac{$THz_s-AP$} \textcolor{blue}{$(\rho_{1})$} and \ac{$THz_c-AP$} \textcolor{blue}{$(1-\rho_{1})$ with respect to users.}}
    \label{fig:5}
    \end{minipage}\hfill
    \begin{minipage}[t]{0.45\textwidth}
     \centering  \includegraphics[width=1.05\linewidth,height=0.78\linewidth]{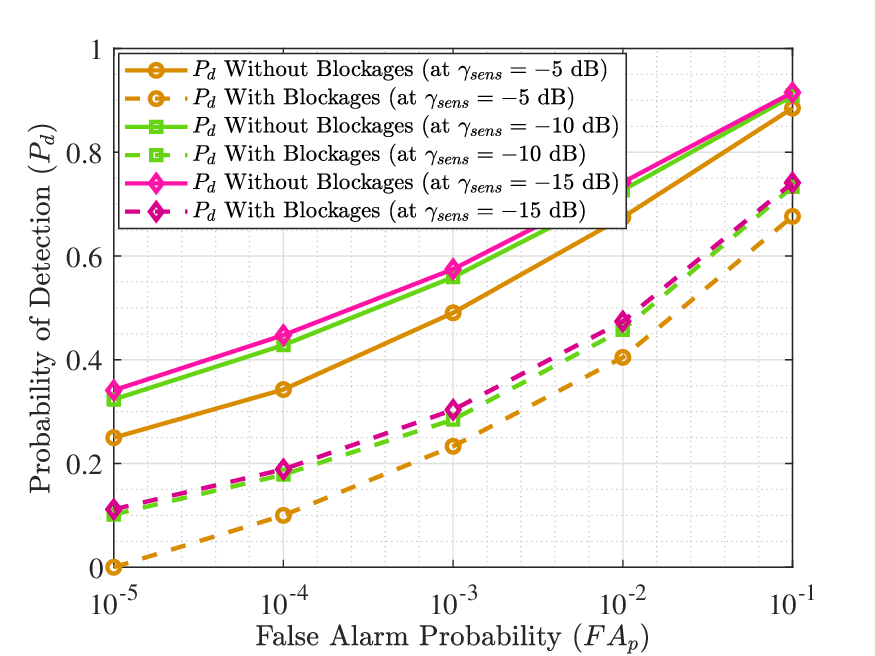}
    \caption{Probability of detection \textcolor{blue}{total $N = 10$} users with respect to the probability of false alarm in the presence and absence of blockages with different $(\gamma_{sens})$.}
    \label{fig:6}
    \end{minipage}
\end{figure}

Figs.3(a) and 3(b) present the \ac{$P_d$} achieved by each user within the sensing coverage region with and without blockages. A detection threshold of $P^{th}_{d}$
 $= 0.5$ is employed as the criterion for access point association, ensuring a minimum level of reliable user detection. \textcolor{blue}{If the $n^{th}$ user with {$P_{d(n)}^{(B)}$}} $> 0.5$ are associated with the \ac{$THz_c-AP$}, thereby benefiting from the high-resolution sensing capabilities of the THz band. In contrast, users with \textcolor{blue}{{$P_{d(n)}^{(B)}$}}$\leq 0.5$ are offloaded to the \ac{$VLC_c-APs$} that provides the highest \ac{SNR}, thus maintaining satisfactory communication performance despite reduced sensing reliability. The attenuation and scattering introduced by blockages lead to a noticeable reduction in the $P_d$ for several users. This degradation in sensing reliability restricts the users that can be effectively supported by the \ac{$THz_c-AP$}. 

Figs. 3(c) and 3(d) illustrate the spectral efficiency achieved by each user in the absence and presence of human blockages. User association with access points is primarily governed by the \ac{$P_d$}. Since the \ac{$THz_c-AP$} provides substantially higher \ac{SE} compared to the \ac{$VLC_c-APs$}, the overall system performance improves when a larger fraction of users are associated with the \ac{$THz_c-AP$}, as long as an adequate \ac{SNR} is maintained for all users. This outcome underscores a key comparative insight in the absence of blockages, reliable sensing and high \ac{$P_d$} allow most users to benefit from superior \ac{$THz_c-AP$} links, whereas under blockages, more users are forced to rely on \ac{$VLC_c-APs$} links, resulting in degraded \ac{SE}.
\begin{figure}[htbp]
   \begin{minipage}[t]{0.45\textwidth}
    \centering  \includegraphics[width=1.05\linewidth,height=0.78\linewidth]{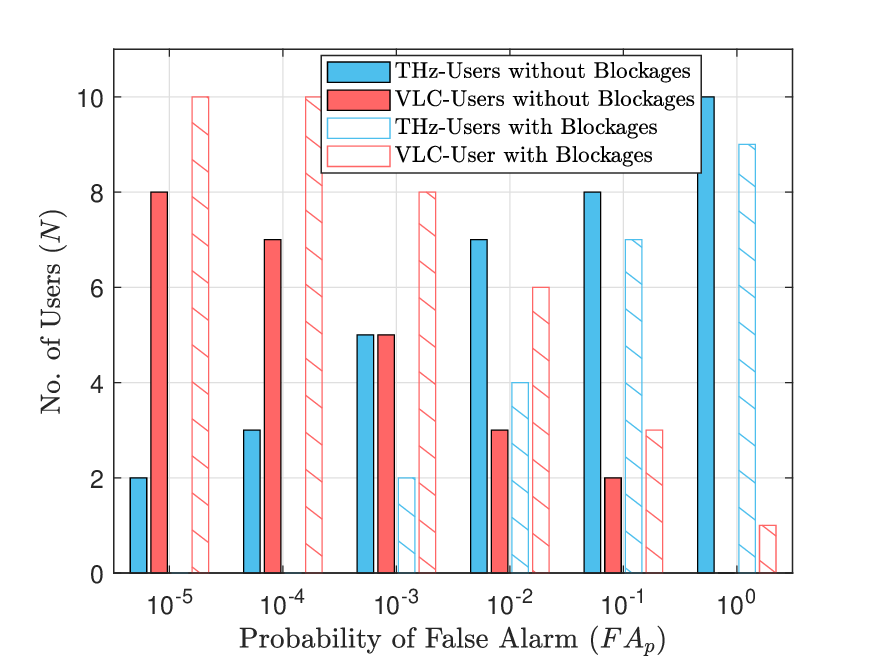}
    \caption{Number of users association from \ac{$THz_c/VLC_c-AP$} with respect to the probability of false alarm \textcolor{blue}{in} the absence and presence of blockages.}
    \label{fig:7}
    \end{minipage}\hfill
    \begin{minipage}[t]{0.45\textwidth}
     \centering    \includegraphics[width=1.05\linewidth,height=0.78\linewidth]{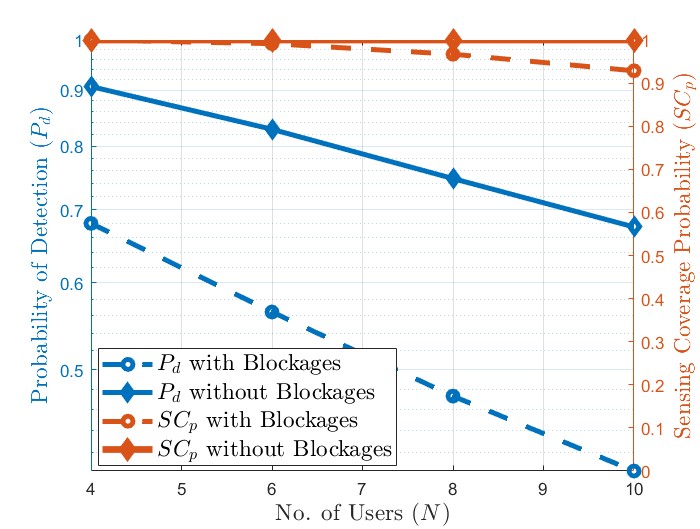}
    \caption{Probability of detection and sensing coverage probability with respect to the number of users \textcolor{blue}{in the absence and presence of human blockages.}}
    \label{fig:8}
    \end{minipage}
\end{figure}

Fig.~\ref{fig:5} illustrates the power allocation factor $\rho_1$ between the sensing \ac{$THz_s-AP$} and the communication \ac{$THz_c-AP$}. The performance varies with different power splits, highlighting the inherent trade-off between power allocation for \ac{THz} sensing and communication. Since the system operates under a fixed total transmit power of $2$~W \textcolor{blue}{with addition circuitry power of 5.6 mW \cite{nguyen2025energy}}, any increase in the fraction $\rho_1$ allocated to the sensing \ac{$THz_s-AP$} inevitably reduces the remaining fraction $1-\rho_1$ available for the communication \ac{$THz_c-AP$}, and vice versa. We also present the results for transmit powers of $P_{\text{w}} = 4$~W and $6$~W \textcolor{blue}{with circuitry power}, where a clear trade-off emerges between THz sensing and communication. This trade-off stems from the fixed total power constraint, where the sensing and communication powers jointly equal $P_{\text{w}}$. Thus, increasing the power allocated to sensing reduces the communication \ac{SNR}, whereas decreasing the sensing power enhances communication performance.

Another key parameter influencing sensing performance is the \ac{$FA_p$}. In Fig.~\ref{fig:6}, we evaluate the \ac{$P_d$} analysis in the absence and presence of blockages, where \ac{$P_d$} is a function of \ac{$FA_p$} under different sensing $SNR$ for N = 10. The results show that \ac{$P_d$} increases with \ac{$FA_p$}, indicating a direct relationship between the two because a sensing detector declares a signal present when its measurement exceeds a threshold. Decreasing the \ac{$FA_p$} from $ 10^ {0} $ to $ 10^ {-5} $ results in more users connecting to the \ac{$VLC_c-APs$}. Conversely, increasing the \ac{$FA_p$}, the more the \ac{$P_d$} is used to detect the true signals. Moreover, the achievable \ac{$P_d$} and \ac{$FA_p$} are significantly influenced by the sensing $SNR$, with higher sensing $SNR$ enabling improved detection performance and \ac{$FA_p$}.

Fig.~\ref{fig:7} illustrates user connectivity to the \ac{$THz_c-AP$} and \ac{$VLC_c-APs$} under both blockage and non-blockage scenarios, evaluated with respect to the \ac{$FA_p$}. In the absence of blockages, an increase in \ac{$FA_p$} results in more users being connected to the \ac{$THz_c-AP$}, while fewer users remain associated with the \ac{$VLC_c-APs$}. However, in the presence of blockages, the situation reverses: a larger fraction of users connect to the \ac{$VLC_c-APs$}, since blockages severely degrade the \ac{$THz_s-AP$} link quality and sensing performance. Consequently, as \ac{$FA_p$} increases under blockage conditions, the users connected through the \ac{$VLC_c-APs$} grows, highlighting the strong interplay between the blockage effects, sensing reliability, and user association.

Fig.~\ref{fig:8} depicts the relationship between the $P_d$ and the $SC_p$ of the users. It can be observed that not all sensed users at $SC_p$ = $1$  are successfully detected \ac{$P_d$} = $0.69$ at $N = 10$, primarily due to the influence of \ac{$FA_p$} $= 10^{-2}$. As a result, while $SC_p$ may approach 100\%, the effective $P_d$ decreases as the number of users increases. Moreover, $SC_p$ is more susceptible to degradation in the presence of blockages, as $SC_p$ = $0.92$ where \ac{$P_d$} = $0$ at $N = 10$, with a higher number of users. Correspondingly, $P_d$ experiences an even sharper reduction under blockage conditions, highlighting the compounded impact of user density and blockages on sensing reliability.
\begin{figure}[htbp]
   \begin{minipage}[t]{0.45\textwidth}
    \centering  \includegraphics[width=1.05\linewidth,height=0.78\linewidth]{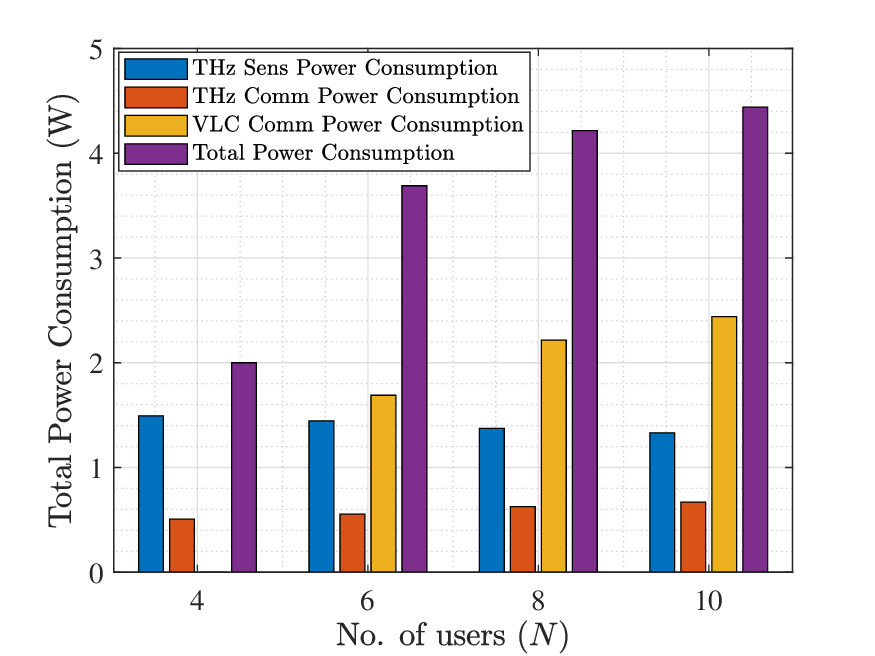}
    \caption{Power consumption \textcolor{blue}{of the \ac{$THz_s-AP$}, \ac{$THz_c-AP$}, and \ac{$VLC_c-APs$} with respect to} number of users.}
    \label{fig:9}
    \end{minipage}\hfill
    \begin{minipage}[t]{0.45\textwidth}
     \centering    \includegraphics[width=1.05\linewidth,height=0.78\linewidth]{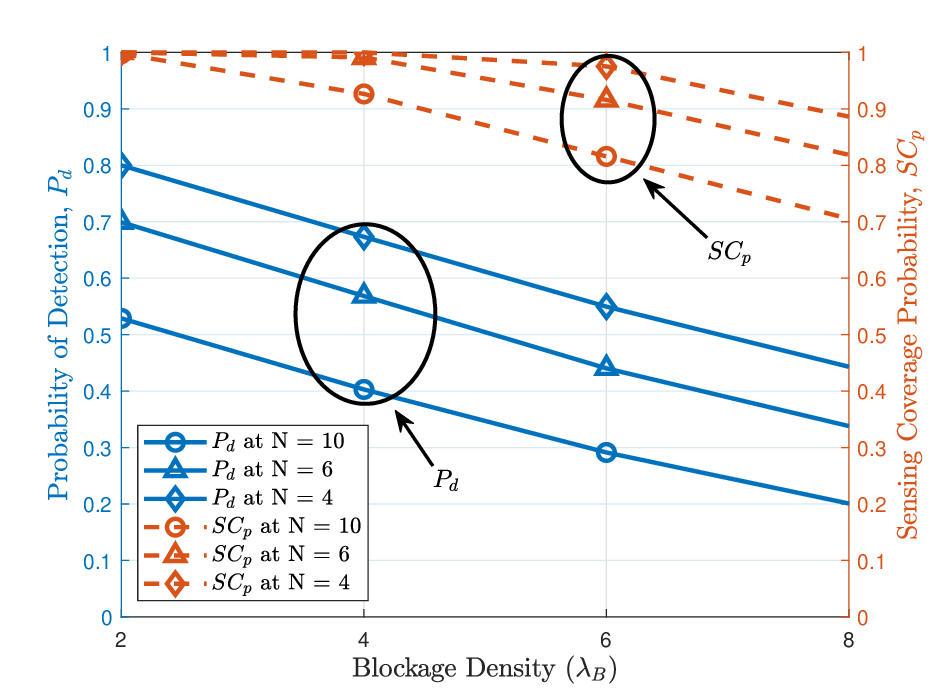}
    \caption{Probability of detection and sensing coverage probability \textcolor{blue}{of different number of users} with respect to the blockage density.}
    \label{fig:10}
    \end{minipage}
\end{figure}

Fig.~\ref{fig:9} illustrates the total power consumption of the proposed system, obtained due to the MILP optimization. A considerable reduction in overall power consumption in the system, which is achieved by solving the MILP optimization, where the transmit powers of the \ac{$THz_s-AP$} and \ac{$THz_c-AP$} are optimized first, followed by adaptive activation of the \ac{$VLC_c-AP$}, each constrained by a maximum transmit power of $5$~W. Only the \ac{$VLC_c-APs$} serving associated users remain active, while the others are switched off to eliminate unnecessary energy usage. This adaptive strategy ensures energy-efficient operation while maintaining user \ac{QoS}. Nevertheless, as the number of users increases, additional \ac{$VLC_c-APs$} become active, thereby increasing the system's power consumption from $2\,\mathrm{W}$ to $4.3\,\mathrm{W}$.

Fig.~\ref{fig:10} illustrates how the \ac{$P_d$} and the $SC_p$ are affected by human blockages. As the density of human blockages increases from 2 to 8, the  \ac{$P_d$} decreases from 0.7 to 0.35, and \ac{$SC_p$} decreases from 1 to 0.9, exhibiting a downward trend. This means that the system's ability to sense and correctly detect users becomes weaker as more obstacles are present. In particular, a higher blockage density significantly reduces the probability of correctly detecting a user, thereby lowering \ac{$P_d$}. The results are obtained while keeping the \ac{$FA_p$} fixed at $10^{-2}$, ensuring that the observed degradation is solely due to the impact of blockages.
\begin{figure}[htbp]
   \begin{minipage}[t]{0.45\textwidth}
    \centering  \includegraphics[width=1.05\linewidth,height=0.78\linewidth]{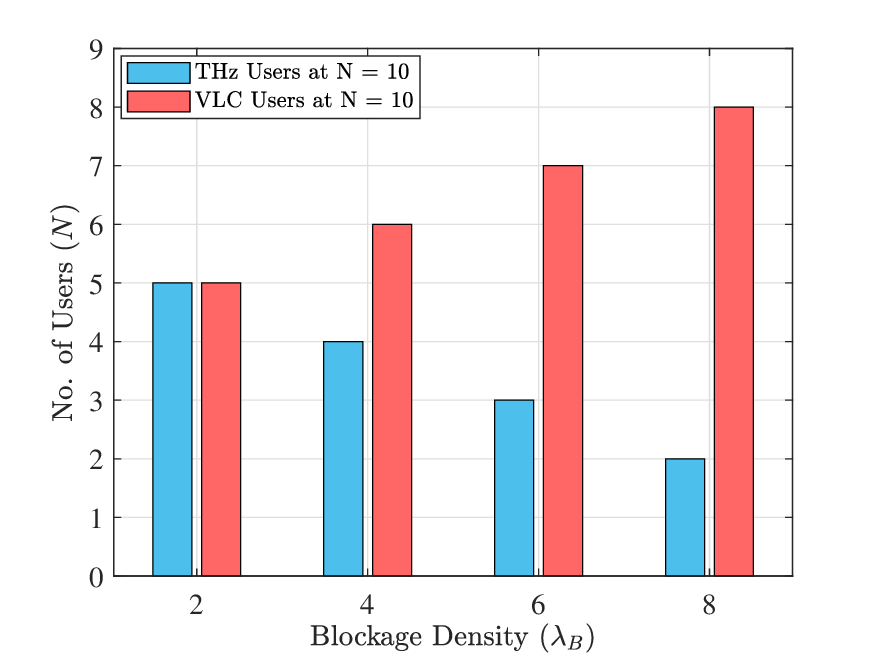}
    \caption{Number of users \textcolor{blue}{(N = 10)} association from \ac{$THz_c/VLC_c-AP$} with respect to Blockage density. }
    \label{fig:11}
    \end{minipage}\hfill
    \begin{minipage}[t]{0.45\textwidth}
     \centering    \includegraphics[width=1.05\linewidth,height=0.78\linewidth]{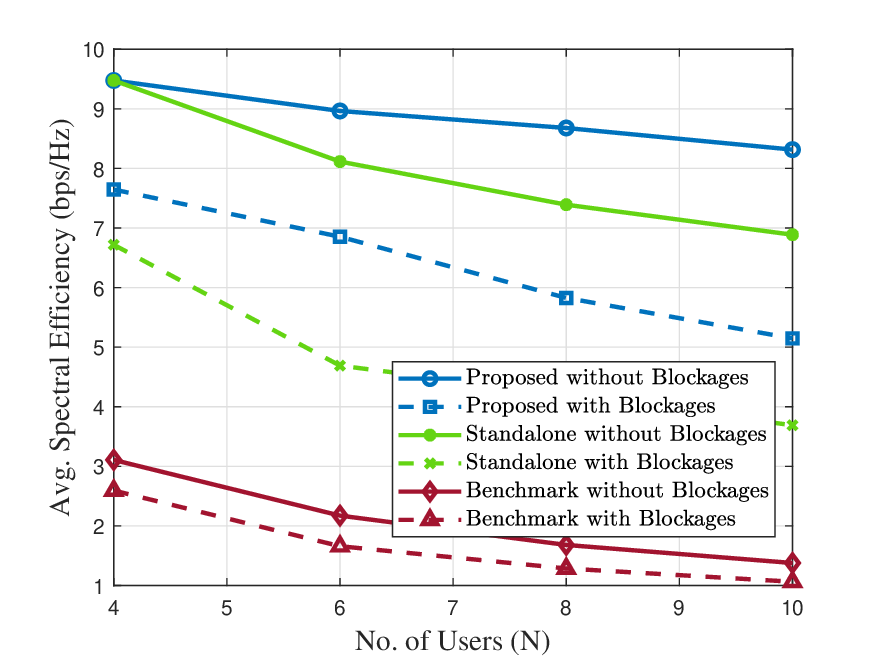}
    \caption{\textcolor{blue}{Average spectral efficiency comparison of the proposed model with the standalone and benchmark models with respect to the number of users.}}
    \label{fig:12}
    \end{minipage}
    \end{figure}
    
As shown in Fig.~\ref{fig:11}, both \ac{$P_d$} and $SC_p$ decrease with increasing blockage density, which directly affects the users ability to maintain reliable connections with the \ac{APs}. Specifically, as the blockage density increases from 2 to 8, the \ac{$P_d$} decreases and fails to connect to via \ac{$THz_c-AP$}, then the larger fraction of users tend to associate with the \ac{$VLC_c-AP$}s instead of the \ac{$THz_c-AP$}.

Fig.~\ref{fig:12} presents the average \ac{SE} comparison between the proposed, the standalone, and the benchmark model. As the $N$ increases, the standalone \ac{$THz_s-AP$} sensing and \ac{$THz_c-AP$} communication setup fails to ensure full coverage, since many users cannot maintain reliable connections with the \ac{APs} due to the high sensitivity of the THz links to blockages. In such scenarios, the \ac{$VLC_c-AP$}s play a crucial role by providing alternative connectivity, thereby maintaining the overall system performance. Consequently, the proposed model achieves significantly higher average \ac{SE} around $8.3$ bps/Hz for the $N = 10$ users than benchmarks, which stands with only $0.5$ bps/Hz. Even under blockage conditions, it effectively exploits the complementary nature of \ac{$THz_c-AP$} and \ac{$VLC_c-APs$} links for robust user association and resource utilization. Furthermore, the hybrid design ensures improved fairness across users by adaptively activating only the serving \ac{AP}, leading to more efficient power usage. The joint operation of \ac{THz} and \ac{VLC} links also enhances physical-layer security and reliability, making the system well-suited for dense indoor 6G scenarios.
\begin{figure}
    \centering \includegraphics[width=0.5\linewidth]{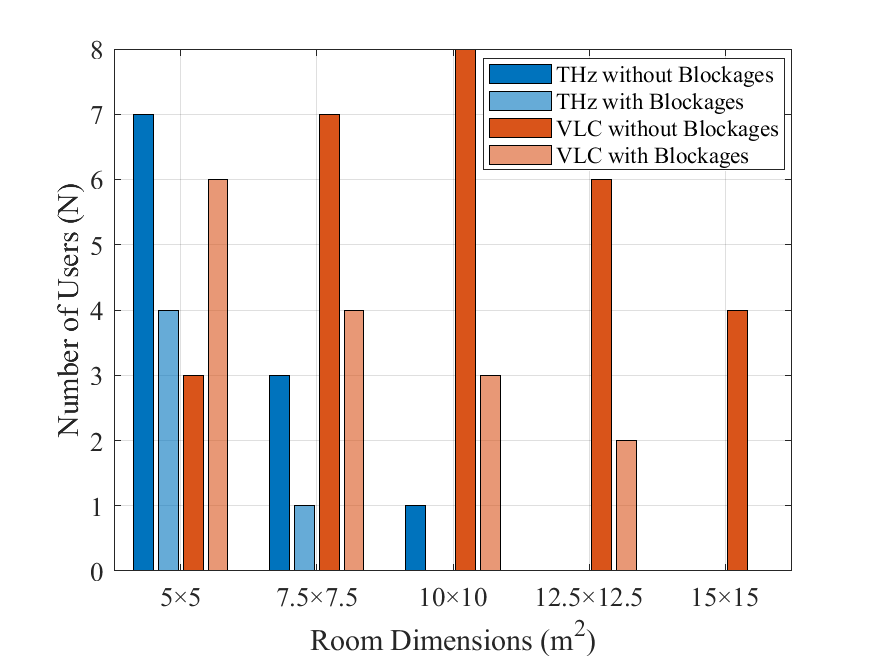}
    \caption{\textcolor{blue}{{Number of users connected to THz/VLC-APs with respect to different room dimensions in the presence and absence of human
blockages.}}}
    \label{fig:13}
\end{figure}

{Fig.~\ref{fig:13} illustrates that for $N = 10$, as the room dimensions increase, the sensing coverage of the \ac{$THz_s-AP$} decreases, resulting in fewer users being detected. Consequently, only a limited number of users can establish a communication link with the \ac{$THz_c-AP$}. For the \ac{$VLC_c-AP$}, communication is provided by the access point offering the highest received signal strength. However, as the room size increases, VLC communication coverage degrades, leading to fewer users being served. In the presence of blockage, the performance of the \ac{$THz_s-AP$} deteriorates significantly; for larger room sizes, almost no users can establish a communication link with the \ac{$THz_c-AP$}. Also under the impact of blockages, only a very small number of users can be supported by the \ac{$VLC_c-AP$}.}
\begin{figure}[htbp]
   \begin{minipage}[t]{0.45\textwidth}
    \centering  \includegraphics[width=1.05\linewidth,height=0.78\linewidth]{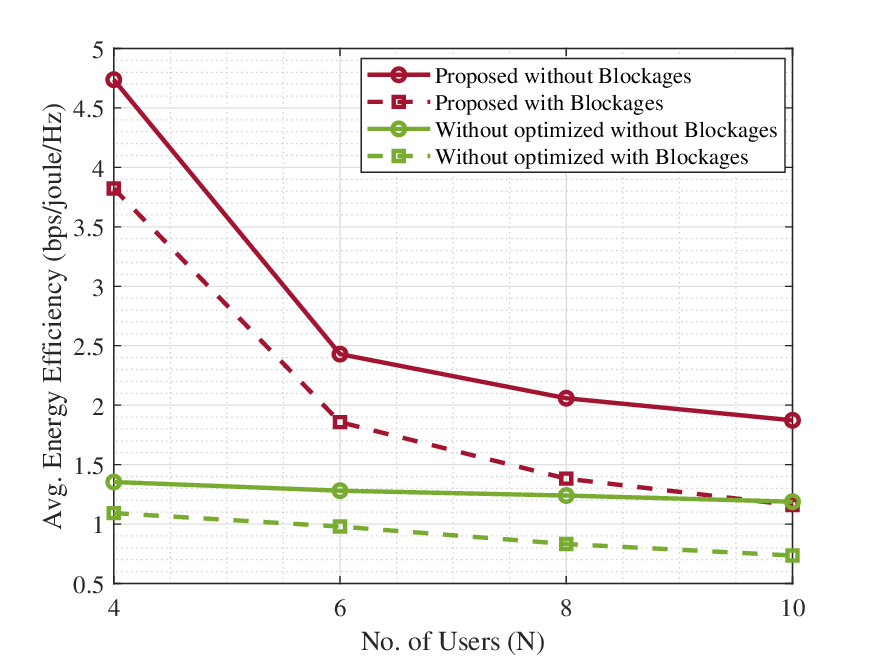}
    \caption{\textcolor{blue}{Proposed model comparison optimized for average energy efficiency with non-optimized respect to the number of users in the presence and absence of human blockages.}}
    \label{fig:14}
    \end{minipage}\hfill
    \begin{minipage}[t]{0.45\textwidth}
     \centering    \includegraphics[width=1.05\linewidth,height=0.78\linewidth]{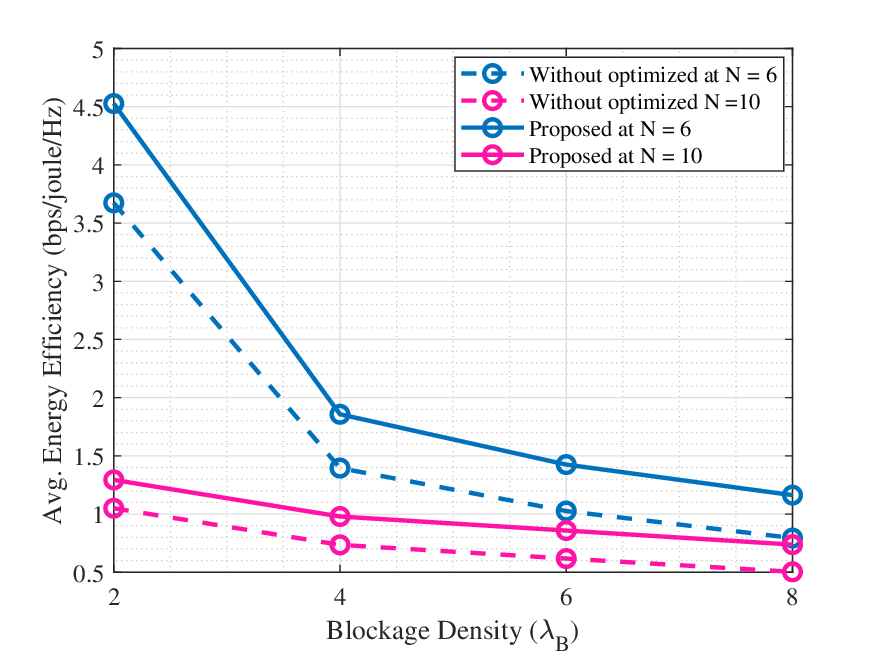}
    \caption{\textcolor{blue}{Proposed model comparison of optimized for average energy efficiency with non-optimized for different numbers of users with respect to the blockage density.}} 
    \label{fig:15}
    \end{minipage}
    \end{figure}
\begin{figure}[htbp]
   \begin{minipage}[t]{0.45\textwidth}
    \centering  \includegraphics[width=1.05\linewidth,height=0.78\linewidth]{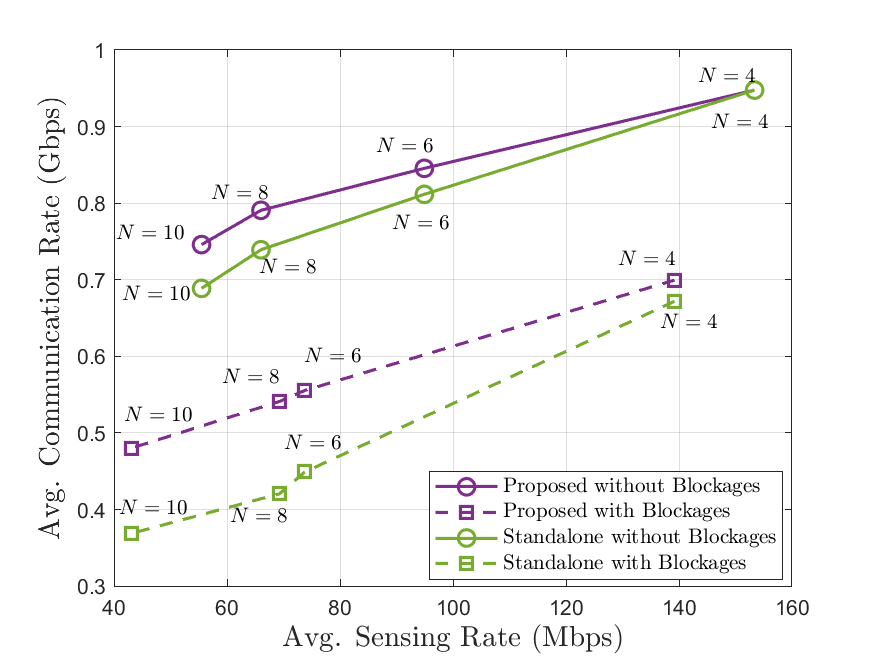}
    \caption{\textcolor{blue}{Comparison analysis of average communication rate with respect to the average sensing rate for different numbers of users in the presence and absence of human blockages.}}
    \label{fig:16}
    \end{minipage}\hfill
    \begin{minipage}[t]{0.45\textwidth}
     \centering    \includegraphics[width=1.05\linewidth,height=0.78\linewidth]{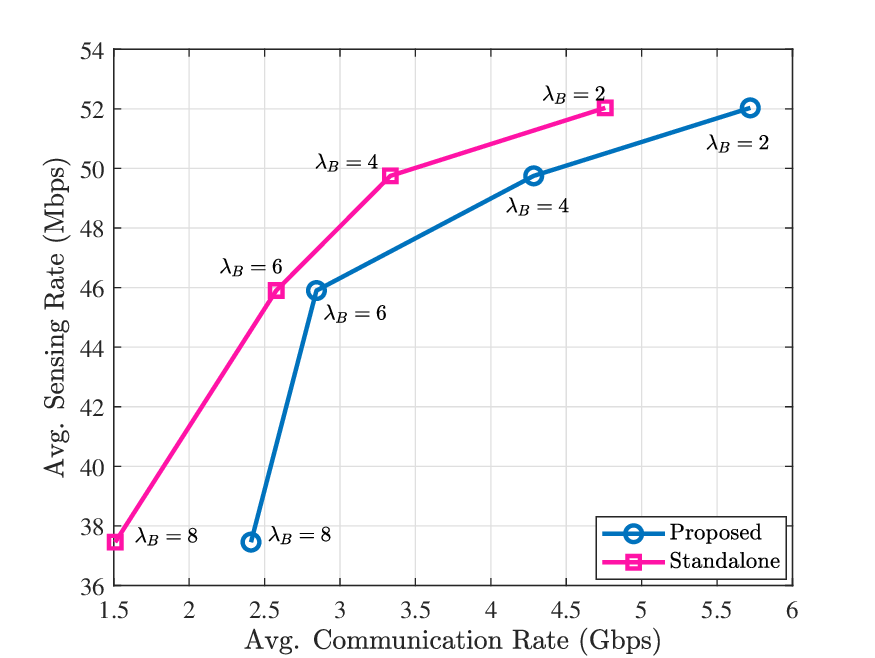}
    \caption{\textcolor{blue}{Comparison analysis of average communication rate with respect to the average sensing rate for different blockage densities.}}
    \label{fig:17}
    \end{minipage}
    \end{figure}
    
Fig.~\ref{fig:14} results demonstrate that the optimized proposed model achieves a higher \ac{EE} of 4.7 bps/joule/Hz for $N = 4$ users compared to the non-optimized baseline of 3.9 bps/joule/Hz, primarily due to its ability to operate with lower transmit power while maintaining communication quality. Moreover, by intelligently leveraging the hybrid THz/VLC architecture, the system effectively mitigates the impact of human blockages, ensuring reliable connectivity even in dynamic environments. Consequently, the proposed framework delivers superior \ac{EE} and robustness compared to conventional non-optimized schemes.

In Fig.~\ref{fig:15} it is shown that as $\lambda_B$ increases from 2 to 8, signal degradation becomes more severe, typically forcing non-optimized systems to increase their transmit power, which leads to a notable drop in \ac{EE}. In contrast, the optimized model adapts effectively to varying blockage conditions by intelligently allocating resources and avoiding unnecessary power consumption. As a result, the optimized scheme sustains higher energy efficiency and more stable performance even under dense blockage scenarios, whereas the non-optimized baseline suffers a significant efficiency decline under the same conditions.

Fig.~\ref{fig:16} shows the variation of communication and sensing rates with the number of users in the absence and presence of blockages. The proposed model is compared against a standalone scheme. As the number of users increases from $N = 4$ to $N = 10$, both sensing and communication rates decline due to resource contention and blockage effects, with the degradation becoming more severe at higher user densities where reliable \ac{LoS} links are harder to maintain. Importantly, the proposed model consistently outperforms the standalone scheme by providing higher sensing and communication rates of 100 Mbps and 0.85 Gbps, respectively, for N = 6 in the absence of blockages. Similarly, 70 Mbps and 0.55 Gbps in the presence of blockages. This gain is achieved through the joint exploitation of \ac{THz} and \ac{VLC} links, whose complementary characteristics ensure robust connectivity and improved system efficiency even in the blockage-prone conditions.

Fig.~\ref{fig:17} presents the average communication and average sensing rates for $N = 10$ users under different blockage densities $(\lambda_B)$. A trade-off exists between the two modalities due to their different distance-dependent behaviors: communication performance typically decreases proportionally to $1/d^{2}$, while sensing performance drops faster, proportional to $1/d_{f}^{4}$ given in equation \eqref{eq:7}. As a result, a reduction in sensing rate from $52$~Mbps to $38$~Mbps can correspond to a relative increase in communication rate from $2.5$~Gbps to $5.8$~Gbps. However, when $\lambda_{B}$ increases, both sensing and communication rates decline simultaneously due to reduced \ac{LoS} availability. Despite this, the proposed model consistently achieves higher communication rates than the standalone scheme by leveraging the complementary strengths of VLC and THz links to maintain robust connectivity and enhance overall efficiency, even in dense blockage scenarios. 

{\section{discussion and limitations }
The proposed framework leverages the complementary strengths of \ac{THz} and \ac{VLC} technologies to enhance indoor wireless performance. THz links support ultra-high data rates and low latency for applications, but are highly susceptible to blockage. \ac{VLC} provides robust short-range connectivity with built-in illumination, enabling reliable service continuity during THz link interruptions. The joint power-allocation and access-point-activation strategy further improves energy efficiency by leveraging existing lighting infrastructure. Overall, the proposed framework enhances spectral, energy, and reliability efficiency in blockage-prone indoor \ac{ISAC} environments.
Despite its advantages, the proposed system has few limitations. The current model assumes static users and perfect channel state information (CSI), whereas user mobility and imperfect CSI in practical deployments may affect link reliability, handover frequency, and system latency. Moreover, \ac{THz} and \ac{VLC} performance is inherently limited by \ac{LoS} requirements, ambient light interference, and user orientation, which can restrict coverage and achievable data rates. Finally, multi-user interference scenarios are not considered. Extending the framework to dynamic multi-user and multi-cell scenarios with adaptive resource allocation is left for future work.}

\section{Conclusion}
{In this paper, we address the problem of power consumption reduction in indoor wireless access networks by considering a hybrid system comprising a \ac{$THz_s-AP$} for user sensing and \ac{$THz_c-AP$}/{$VLC_c-APs$} for communication under the impact of human blockages. Unlike prior studies, we exploit \ac{THz}-based user sensing integrated with hybrid \ac{THz}/\ac{VLC} communication and employ MILP-based power optimization to significantly reduce overall power consumption. The proposed framework jointly integrates \ac{THz} sensing and hybrid \ac{THz}/\ac{VLC} communication access technologies and formulates a power minimization problem that satisfies user communication requirements in the presence of human blockages. Furthermore, a comprehensive performance evaluation of the system energy efficiency (\ac{EE}) is presented. We also analyze the trade-off between communication and sensing rates across various scenarios, including standalone operation with and without blockages, as well as the proposed hybrid configuration. Finally, the impact of varying blockage densities on this trade-off is investigated to provide deeper insights into practical system design. Our simulation results demonstrate that the proposed system significantly reduces power consumption and enhances spectral efficiency (SE) to 9.5 bps/Hz, compared to the conventional benchmark of 3 bps/Hz without blockage. Furthermore, the energy efficiency (EE) improves by 4.7 bps/J/Hz in the absence of blockage and by 3.8 bps/J/Hz in its presence.}

For future research, we aim to examine diverse channel models and explore application-oriented scenarios of \ac{THz}-based sensing integrated with hybrid communication architectures. Furthermore, by incorporating machine learning (ML) techniques for blockage prediction, the system can proactively adapt user association and resource allocation to maintain reliable connectivity. This will offer deeper insights into the practical deployment and performance trade-offs of such systems.
\bibliographystyle{IEEEtran}
\bibliography{reff}
\end{document}